\documentclass[superscriptaddress,aps,pra,twocolumn,nofootinbib]{revtex4-1}
\usepackage{etex}
\usepackage{amsmath,amssymb,amsthm}
\usepackage{easybmat}
\usepackage{hyperref}
\usepackage[pdftex]{graphicx}
\usepackage{times,txfonts}
\usepackage{braket}
\usepackage{color}
\usepackage{natbib}
\newcommand{\be}{\begin{equation}}
\newcommand{\ee}{\end{equation}}
\newcommand{\ba}{\begin{eqnarray}}
\newcommand{\ea}{\end{eqnarray}}
\newcommand{\ketbra}[2]{|#1\rangle \langle #2|}

\newcommand{\half}{\frac{1}{2}}
\newcommand{\qua}{\frac{1}{4}}
\newcommand{\etal}{{\it{et al. }}}

\newtheorem{observation}{Observation}
\newtheorem{definition}{Definition}
\newtheorem{theorem}{Result}
\begin{document}
\title{Canonical decomposition of quantum correlations in the framework of generalized nonsignaling theories}   
   \author{C. Jebarathinam}
\email{jebarathinam@gmail.com; jebarathinam@iisermohali.ac.in}
\affiliation{Indian Institute of Science Education and Research (IISER) Mohali, \\
Sector-81, S.A.S. Nagar, Manauli 140306, India.}
\date{\today}
\begin{abstract}
We introduce the measures, Bell discord (BD) and Mermin discord (MD), to  
characterize bipartite quantum correlations in the context of nonsignaling (NS) polytopes.
These measures divide the full NS polytope into four regions 
depending on whether BD and/or MD is zero. This division of the NS polytope 
allows us to obtain a $3$-decomposition that any bipartite box with two binary inputs and two binary outputs can be decomposed into 
Popescu-Rohrlich (PR) box, a maximally local box, and a local box with BD and MD equal to zero.
BD and MD quantify two types of nonclassicality of correlations arising from all quantum correlated states 
which are neither classical-quantum states nor quantum-classical states.
BD and MD serve us the semi-device-independent witnesses of nonclassicality of local boxes in that nonzero value of these measures imply   
incompatible measurements and nonzero quantum discord only when the dimension of the measured states is fixed.  
The $3$-decomposition serves us to isolate the origin of the two types of nonclassicality into a PR-box and a maximally local box which is related 
to EPR-steering, respectively. We study a quantum polytope that has an overlap with all the four regions of the full NS polytope to figure
out the constraints of quantum correlations.
\end{abstract}
\maketitle
\section{Introduction}
Bell showed that measurements on spatially separated entangled system can lead to nonlocal correlations
which cannot be explained by local hidden variable (LHV) theory \cite{bell64,BNL}. Nonlocality is witnessed 
by the violation of a Bell inequality which puts an upper bound on the correlations under the constraint 
of the LHV theory. Nonlocality is not the unique nonclassical feature of quantum theory as nonquantum 
correlations, which cannot be used for instantaneous signaling, also violate a Bell inequality. 
Nonlocality of quantum theory is further limited by the Tsirelson bound \cite{tsi1}. 
Popescu and Rohrlich showed that the limited violation of a Bell inequality by quantum theory is not a consequence of relativity \cite{PR}. 

In generalized nonsignaling theory (GNST), correlations are constrained only by the nonsignaling (NS) principle \cite{Barrett,MAG06}. 
The set of NS correlations forms a convex polytope which can be divided into a nonlocal region and local polytope. 
The set of quantum correlations forms a convex set, however, it is not a polytope \cite{Pitowski}. Since quantum correlations are contained in the NS polytope, 
any quantum correlation can be written as a convex combination of the extremal boxes of the polytope. Thus, quantum correlations can be studied using NS polytope. 
One of the goals of studying GNST is to find out what singles out quantum theory from other nonsignaling theories \cite{Geb}. 
Quantum key distribution was studied in the context of GNST in which nonlocality is responsible for security \cite{DQKD}.

Nonlocality of quantum theory implies that nonlocal correlations are obtained by incompatible measurements performed on the entangled states \cite{BNL}. 
All pure bipartite entangled states violate a Bell inequality for appropriate incompatible measurements \cite{GT,PRQB}. 
However, Werner showed that nonlocality and entanglement are inequivalent; there are mixed entangled states which have LHV models for all measurements \cite{Werner}. 
Thus, not all entangled states can lead to the violation of a Bell inequality even when incompatible measurements are performed on them \cite{IncomN}.
Nonlocality is one of the manifestations of contextuality, where commutation of the observables comes from the fact that, they
are spatially separated \cite{Grudkaetal}. Quantum contextuality is manifested through logical contradiction known as Kochen-Specker (KS) paradox \cite{KS}. 
Peres \cite{Peres} showed that for a certain choice of incompatible measurements, the maximally entangled state exhibits KS paradox without nonlocality \cite{UNLH}. 
Quantum discord was introduced as a measure of 
quantum correlations which quantifies nonclassicality of separable states as well \cite{WZQD}. It is interesting to study correlations 
arising from incompatible measurements performed on the nonzero quantum discord states in the context of GNST. 

In this work, we introduce the measures, Bell discord and Mermin discord, to characterize quantum correlations 
in the framework of GNST. We restrict to the NS polytope in which the black boxes have two binary inputs and two binary outputs.
We characterize only those NS boxes with two binary inputs and two binary outputs.
Bell and Mermin discord are analogous to geometric measure of quantum discord \cite{Dakicetal}; just as quantum discord quantifies nonclassicality of separable states, 
these two measures quantify nonclassicality of correlations admitting LHV model as well, and they are geometric measures in the NS polytope \cite{Jeba}. 
Bell discord is constructed using Bell-CHSH operators \cite{chsh}, whereas Mermin discord is constructed using Mermin operators \cite{mermin}, hence the names. 
%We obtain canonical decomposition of any NS correlation with respect to Bell and Mermin discord.
The extremal nonclassical correlations with respect to Bell discord are the extremal nonlocal boxes (also known as PR-boxes), 
whereas the extremal nonclassical correlations 
with respect to Mermin discord are maximally local boxes which we call Mermin boxes.
We show that any quantum correlations can be decomposed into PR-box, a Mermin box and a purely classical box;
Bell discord and Mermin discord quantify the PR box and Mermin box components, respectively, in this $3$-decomposition.
We show that the correlations arising from quantum-correlated states 
which have nonzero left and right quantum discord \cite{Dakicetal,QQC} can have a $3$-decomposition 
for suitable choice of incompatible measurements.

%We find that the Mermin boxes violate an EPR steering inequality. 
%We show that the correlations arising from all quantum-correlated states \cite{QQC}, which have nonzero left and right quantum discord \cite{Dakicetal}, 
%can have nonzero Bell discord or Mermin discord or both of them together when appropriate incompatible measurements are performed on them. 

The paper is organized as follows. In Sec. \ref{prl}, we discuss the motivations for defining nonclassicality for the Bell-local correlations. 
In Sec. \ref{NSp}, we review the geometry of bipartite 
nonsignaling boxes. 
In Sec. \ref{BDaMD}, we define the measures of Bell discord and Mermin discord, and we find the $3$-decomposition fact of NS boxes. 
In sec. \ref{QC},  we analyze quantum correlations in pure states, Werner states, mixture of maximally entangled state with classically correlated state, 
and classical-quantum and quantum-classical states using these measures. Conclusions are presented in Sec. \ref{conc}. 
\section{Preliminaries}\label{prl}
Consider the Bell-CHSH scenario \cite{chsh} in which two spatially separated parties have access to subsystems of a bipartite system and make two dichotomic measurements $A_i$ 
and $B_j$ on their respective subsystems which produces binary outcomes $a_m$ and $b_n$; 
the correlation between the outcomes is described by the conditional joint probability distributions, $P(a_m,b_n|A_i,B_j)$, 
which can be represented in matrix notation as follows,
\begin{widetext}
\begin{equation}
\left( \begin{array}{cccc}
P(a_0,b_0|A_0,B_0) & P(a_0,b_1|A_0,B_0) & P(a_1,b_0|A_0,B_0) & P(a_1,b_1|A_0,B_0) \\
P(a_0,b_0|A_0,B_1) & P(a_0,b_1|A_0,B_1) & P(a_1,b_0|A_0,B_1) & P(a_1,b_1|A_0,B_1) \\
P(a_0,b_0|A_1,B_0) & P(a_0,b_1|A_1,B_0) & P(a_1,b_0|A_1,B_0) & P(a_1,b_1|A_1,B_0) \\
P(a_0,b_0|A_1,B_1) & P(a_0,b_1|A_1,B_1) & P(a_1,b_0|A_1,B_1) & P(a_1,b_1|A_1,B_1) \\
\end{array} \right).
\end{equation}
\end{widetext}
Suppose Alice and Bob generate $P(a_m,b_n|A_i,B_j)$ by making von Neumann measurements on an ensemble of bipartite two-qubit systems described 
by the density matrix $\rho_{AB}$ in the Hilbert space $\mathcal{H}^2_A\otimes\mathcal{H}^2_B$. Quantum theory predicts the correlation through the Born's rule, 
\begin{equation}
P(a_m,b_n|A_i,B_j)=\mathrm{Tr}\left(\rho_{AB}\Pi_{A_i}^{a_m}\otimes\Pi_{B_j}^{b_n}\right),
\end{equation}
where $\Pi^{a_{m}}_{A_i}={1/2}\{\openone+a_{m}\hat{a}_i \cdot \vec{\sigma}\}$ and $\Pi^{b_{n}}_{B_j}={1/2}\{\openone+b_{n}\hat{b}_j \cdot \vec{\sigma}\}$ 
are the projectors generating binary outcomes $a_{m},b_{n} \in \{-1,1\}$. Since quantum correlation arises from the tensor product structure, 
by definition it is nonsignaling; that is the marginal distribution of Alice is independent of the measurement choice made by Bob and vice versa, 
i.e., $P(a_m|A_i)\equiv\sum_n P(a_m,b_n|A_i,B_j)=\mathrm{Tr}\left(\rho_{AB} \Pi^{a_m}_{A_i}\otimes\openone\right)$ and $P(b_n|B_j)\equiv\sum_m P(a_m,b_n|A_i,B_j)
=\mathrm{Tr}\left(\rho_{AB} \openone \otimes \Pi^{b_n}_{B_j}\right)$. 

A nonsignaling box that achieves maximal Bell nonlocality is 
known as PR-box \cite{Barrett}; for instance, 
the canonical PR-box \cite{PR},
\begin{equation}
P_{PR} = \left( \begin{array}{cccc}
\half & 0 & 0 & \half \\
\half & 0 & 0 & \half \\
\half & 0 & 0 & \half \\
0 & \half & \half & 0 
\end{array} \right),
\label{eq:prbox}
\end{equation}
maximally violates the Bell-CHSH inequality \cite{chsh},
\be
\mathcal{B}:=\braket{A_0B_0}+\braket{A_0B_1}+\braket{A_1B_0}-\braket{A_1B_1}\le2, \label{cBCHSH}
\ee
where $\braket{A_iB_j}=\sum_{mn}a_mb_nP(a_m,b_n|A_i,B_j)$. 
We now consider isotropic PR-box \cite{MAG06} which is a mixture of the PR-box and white noise,
\be
P=pP_{PR}+(1-p)P_N. \label{PRiso}
\ee
Here 
\begin{equation}
P_{N} = \left( \begin{array}{cccc}
\qua & \qua & \qua & \qua \\
\qua & \qua & \qua & \qua \\
\qua & \qua & \qua & \qua \\
\qua & \qua & \qua & \qua 
\end{array} \right).
\label{eq:wn}
\end{equation}
The isotropic PR-box violates the Bell-CHSH inequality i.e., $\mathcal{B}=4p>2$ if $p>\frac{1}{2}$.
Notice that even if the isotropic PR-box is local when $p\le \frac{1}{2}$, it 
has the single PR-box component if $p>0$. We call such a single PR-box in the decomposition of any box (nonlocal, or not) irreducible PR-box.

The isotropic PR-box which is quantum physically realizable if $p\le\frac{1}{\sqrt{2}}$ \cite{MAG06} illustrates the following observation.
\begin{observation}
%When incompatible measurements on an entangled state gives rise to Bell nonlocality, the correlation has irreducible PR-box component.
When local boxes arising from two-qubit entangled states have an irreducible PR-box component, the measurements that give rise to them are incompatible
i.e., measurement observables on Alice's and Bob's sides are noncommuting: $[A_0,A_1]\ne0$ and $[B_0,B_1]\ne0$.
%There are local quantum correlations arising from incompatible measurements on the entangled states which has irreducible PR-box component.
\end{observation}
\begin{proof}
Quantum correlation that achieves maximal Bell nonlocality is obtained by making suitable incompatible measurements 
on a maximally entangled state; for instance, the Bell state, $\ket{\psi^+}=\frac{1}{\sqrt{2}}(\ket{00}+\ket{11})$, gives rise to 
the Tsirelson bound i.e., $\mathcal{B}=2\sqrt{2}$ for the measurement observables: $A_0=\sigma_x$, $A_1=\sigma_y$,
$B_0=\frac{1}{\sqrt{2}}(\sigma_x-\sigma_y)$ and $B_1=\frac{1}{\sqrt{2}}(\sigma_x+\sigma_y)$. 
The NS box that achieves the Tsirelson bound can be written in the isotropic PR-box form with the irreducible PR-box content $p=\frac{1}{\sqrt{2}}$
(Tsirelson box):
\be
P=\frac{1}{\sqrt{2}}P_{PR}+\left(1-\frac{1}{\sqrt{2}}\right)P_N. \label{MQNL}
\ee
Since the irreducible PR-box component in this decomposition corresponds to maximal entanglement, 
it is natural to ask whether the isotropic PR-box in Eq. (\ref{PRiso}) 
can arise from nonmaximally entangled states when $0<p<\frac{1}{\sqrt{2}}$.
Indeed, the pure nonmaximally entangled states,
\be
\ket{\psi(\theta)}=\cos\theta\ket{00}+\sin\theta\ket{11},  \label{nmE}
\ee
gives rise to the isotropic PR-box given in Eq. (\ref{PRiso}) 
with $p=\frac{\sin2\theta}{\sqrt{2}}$ for the measurements that gives rise to the Tsirelson box in Eq. (\ref{MQNL}). 
%Notice that the nonmaximally entangled states 
%do not give rise to the Bell-CHSH inequality violation if $\sin2\theta\le\frac{1}{\sqrt{2}}$ and
%the correlations have the irreducible PR-box component whenever the state is entangled. 
\end{proof}
The observation that local boxes which have an irreducible PR-box component can arise from incompatible measurements on entangled states
motivates to define a notion of nonclassicality which we call Bell discord.  
\begin{definition}\label{BDdef1}
A correlation arising from incompatible measurements performed on a given two-qubit state 
has \textit{Bell discord} 
iff it admits a decomposition with an irreducible PR-box component. 
\end{definition}
Bell discord is not equivalent to Bell nonlocality since local boxes can also have an irreducible PR-box component;
for instance, the isotropic PR-box in Eq. (\ref{PRiso}) has Bell discord if $p>0$, whereas it has Bell nonlocality if $p>\frac{1}{2}$.  

EPR-steering is a weaker form of quantum nonlocality in which incompatible measurements on one subsystem of an entangled state 
prepare different ensembles for the other subsystem and is witnessed by the violation of an EPR-steering inequality \cite{EPRsi}.
For the incompatible measurements: 
$A_0=\sigma_x$, $A_1=\sigma_y$, $B_0=\sigma_x$ and $B_1=\sigma_y$,  the Bell state, $\ket{\psi^+}$, does not give rise to Bell nonlocality, 
however, it gives rise to the violation of the following EPR-steering inequality \cite{CJWR},
\be
\braket{A_0\sigma_x}-\braket{A_1\sigma_y}\le \sqrt{2}. \label{eprst}
\ee
For this choice of measurements, 
the Bell state gives rise to the following correlated local box,
\begin{equation}
P_M = \left( \begin{array}{cccc}
\half & 0 & 0 & \half \\
\qua & \qua & \qua & \qua \\
\qua & \qua & \qua & \qua \\
0 & \half & \half &  0
\end{array} \right).
\label{eq:merminbox}
\end{equation}
Notice that the above box
is maximally local in that it gives the local bound of the Bell-CHSH inequality, i.e., $\mathcal{B}=2$.

Just as the PR-box exhibits logical contradiction with local realism,
the correlated box in Eq. (\ref{eq:merminbox}) exhibits the logical contradiction with realistic value
assignment as follows: The first and fourth rows in Eq. (\ref{eq:merminbox}) imply that the outcomes satisfy $A_0B_0=1$ and $A_1B_1=-1$; if the outcomes are predetermined
realistically, it should satisfy, $A_0B_1A_1B_0=-1$, but this contradicts the  rows $2$ and $3$ because there is  a nonzero  probability  
for $A_0B_1=A_1B_0=1$ or  $A_0B_1=A_1B_0=-1$. This argument is inspired by the Peres' version of KS paradox \cite{Peres}.
Notice that the following maximally local and correlated box,
\begin{equation}
P_{CC} = \left( \begin{array}{cccc}
\half & 0 & 0 & \half \\
0 & \half & \half & 0 \\
\half & 0 & 0 & \half \\
0 & \half & \half &  0
\end{array} \right),
\label{eq:ccbox}
\end{equation}
does not exhibit the above logical contradiction.
We call a maximally local and correlated box which exhibits the logical contradiction with the realistic value assignment Mermin box.  

For incompatible measurements that lead to the maximal violation of the EPR-steering inequality in Eq. (\ref{eprst}), the nonmaximally entangled states in Eq. (\ref{nmE})
give rise to isotropic Mermin box which is a convex mixture of the Mermin box in Eq. (\ref{eq:merminbox}) and white noise,
\be
P=p P_M+(1-p)P_N, \label{Mmot}
\ee
with $p=\sin2\theta$. Analogous to the isotropic PR-box,
the isotropic Mermin box arising from the pure entangled states, $\ket{\psi(\theta)}$, violates the EPR-steering inequality  
if $\sin2\theta>\frac{1}{\sqrt{2}}$, however,
it has the irreducible Mermin box component whenever the state is entangled.
%Thus, the nonviolation of the EPR-steering inequality by the correlation  
%does not imply that that the correlation cannot have an irreducible Mermin box component.
%We obtain the following observation from the isotropic Mermin box. 
%\begin{observation}
%When incompatible measurements on an entangled state gives rise to Bell nonlocality, the correlation has irreducible PR-box component.
%When a local correlation arising from a given two-qubit state  has an irreducible Mermin box component, 
%the measurements that give rise to them are incompatible even if the correlation does not violate an EPR-steering inequality.
%There are local quantum correlations arising from incompatible measurements on the entangled states which has irreducible PR-box component.
%\end{observation}
%We call such a single Mermin box in the decomposition of any correlation
%irreducible Mermin box.
The observation that the isotropic Mermin box can arise from incompatible measurements on the entangled states motivates to define a
notion of nonclassicality which we call Mermin discord.
\begin{definition}\label{MDdef}
A correlation arising from incompatible measurements performed on a given two-qubit state 
has \textit{Mermin discord} 
iff it admits a decomposition with an irreducible Mermin box component.  
\end{definition}
We observe that the isotropic Mermin box can exhibit EPR-steering only when the Mermin box component 
is larger than a certain amount. Thus, analogous to the statement that Bell discord and nonlocality are inequivalent, 
we have the observation that Mermin discord is not equivalent to EPR-steering.   

The definitions \ref{BDdef1} and \ref{MDdef} imply that an isotropic PR-box which has Bell discord does not have Mermin discord, 
and an isotropic Mermin box which has Mermin discord does not have  
Bell discord. Thus, Bell discord and Mermin discord are two different nonclassical features of quantum correlations which go beyond
nonlocality; the former originates
from nonlocality, whereas the latter originates from EPR-steering. 

\begin{observation}
Quantum correlations can have Bell and Mermin discord simultaneously. 
\end{observation}
\begin{proof}
For the measurements:
$A_0=\sigma_x$, $A_1=\sigma_y$, $B_0=\sqrt{p}\sigma_x-\sqrt{1-p}\sigma_y$ $\&$ $B_1=\sqrt{1-p}\sigma_x+\sqrt{p}\sigma_y$, where $\frac{1}{2}\le p \le1$, 
the correlations arising from the Bell state, $\ket{\psi^+}$, can be decomposed into PR-box, the Mermin box and white noise,
\be
P=\mu P_{PR}+\nu P_M+(1-\mu-\nu)P_N, \label{MeCNL}
\ee 
where $\mu=\sqrt{1-p}$ and $\nu =\sqrt{p}-\sqrt{1-p}$. These correlations 
have the components of irreducible PR-box and Mermin-box when $\frac{1}{2}< p <1$.
\end{proof}
We call a decomposition of a quantum correlation that has the fractions of PR-box and Mermin-box, and, a part that does not have
Bell and Mermin discord $3$-decomposition, for instance,  
the box arising from the maximally entangled state in Eq. (\ref{MeCNL}) has a $3$-decomposition. 
The isotropic PR-box and Mermin box given by the decompositions in Eqs. (\ref{PRiso}) and (\ref{Mmot}) 
are special instances of the $3$-decomposition in that in these two decompositions 
one of the nonclassical terms is zero.
We will obtain a $3$-decomposition of any NS box by using the geometry of the NS polytope
with respect to the measures of Bell discord and Mermin discord which we will define.
\section{Polytope of nonsignaling boxes}\label{NSp}
Barrett \etal\cite{Barrett} have shown that the set of bipartite nonsignaling boxes ($\mathcal{N}$) with two-binary-inputs-two-binary-outputs 
forms an $8$ dimensional convex polytope with $24$ 
vertices. The vertices (or extremal boxes) of this polytope 
are $8$ PR-boxes,
\be
P^{\alpha\beta\gamma}_{PR}(a_m,b_n|A_i,B_j)=\left\{
\begin{array}{lr}
\frac{1}{2}, & m\oplus n=i\cdot j \oplus \alpha i\oplus \beta j \oplus \gamma\\ 
0 , & \text{otherwise}\\
\end{array}
\right. \label{NLV}
\ee
and $16$ deterministic boxes:
\be
P^{\alpha\beta\gamma\epsilon}_D(a_m,b_n|A_i,B_j)=\left\{
\begin{array}{lr}
1, & m=\alpha i\oplus \beta\\
   & n=\gamma j\oplus \epsilon \\
0 , & \text{otherwise}.\\
\end{array}
\right.
\ee
Here $\alpha,\beta,\gamma,\epsilon\in \{0,1\}$  and $\oplus$ denotes addition modulo $2$. Any NS correlation can be written as a convex sum of the $24$ extremal boxes:
 \be
P(a_m, b_n|A_i,B_j)=\sum^7_{k=0}p_kP^k_{PR}+\sum^{15}_{l=0}q_lP^l_{D};
\sum_kp_k+\sum_lq_l=1,     \label{CHNS}    
\ee
here $k=\alpha\beta\gamma$ and $l=\alpha\beta\gamma\epsilon$.
All the deterministic boxes can be written as the product of marginals corresponding to Alice and Bob, $P_D(a_m,b_n|A_i,B_j)=P_D(a_m|A_i)P_D(b_n|B_j)$, 
whereas the $8$ PR-boxes
cannot be written in product form. Note that unlike the deterministic boxes, the marginals of the PR boxes are maximally mixed: {\it i.e.}, 
$P(a_{m}|A_i)=\frac{1}{2}=P(b_{n}|B_j)$ for all $i,j,m,n$. 
The extremal boxes in a given class are related to each other through local reversible operations (LRO). 
LRO simply relabel the inputs and outputs such that the class of the vertices remain invariant: 
Alice changing her input $i\rightarrow i\oplus 1$, and changing her output conditioned on the input: $m\rightarrow m\oplus\alpha i\oplus\beta$. 
Bob can perform similar operations. 
Thus, the extremal boxes in a given class are equivalent under LRO.
\begin{figure}
\centering
\includegraphics[scale=0.30]{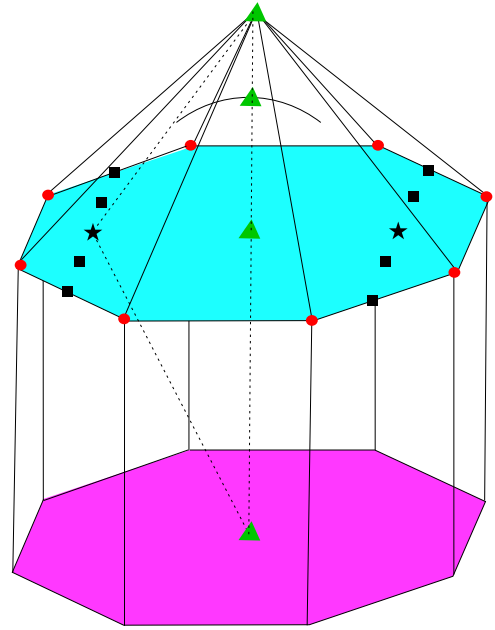} 
\caption{A three-dimensional representation of the NS polytope with two binary inputs and two binary outputs is shown here. The octagonal cylinder  represents the 
local polytope. The lines connecting the deterministic boxes represented by
red points define one of the facet for the local polytope; 
the PR-box which violates the Bell-CHSH inequality corresponding 
to this facet is represented by triangle point on the top of the NS polytope. 
The region below the curved surface contains quantum correlations and the point on this curved surface is the Tsirelson box.
The star and square points on the facet of the local polytope represent maximally and nonmaximally mixed marginals Mermin boxes, respectively. 
The triangular region (shown by dotted lines) which is a convex hull
of the PR-box, the Mermin box and white noise represents the $3$-decomposition fact that  
any point which lies inside the triangle can be decomposed into PR-box,
the Mermin-box and white noise. The line connecting the PR-box and white noise represents the isotropic PR-box and the line connecting the Mermin
box and white noise represents the isotropic Mermin box.}\label{NS3dfig}
\end{figure} 

Bell polytope ($\mathcal{L}$), which is a subpolytope of $\mathcal{N}$, is a convex hull of the $16$ deterministic boxes: if $P(a_m, b_n|A_i,B_j)\in \mathcal{L}$,   
\ba
P(a_m, b_n|A_i,B_j)=\sum^{15}_{l=0}q_lP^l_{D}; \sum_lq_l=1. \label{LD}
\ea
Fine \cite{Fine} showed that a correlation can be simulated by the local hidden variable model in Eq. (\ref{LD}) iff the correlation 
satisfies the complete set of Bell-CHSH inequalities \cite{WernerWolf}:
\ba
\mathcal{B}_{\alpha\beta\gamma} &:= &(-1)^\gamma\braket{A_0B_0}+(-1)^{\beta \oplus \gamma}\braket{A_0B_1}\nonumber\\
&+&(-1)^{\alpha \oplus \gamma}\braket{A_1B_0}+(-1)^{\alpha \oplus \beta \oplus \gamma \oplus 1} \braket{A_1B_1}\le2, \label{BCHSH}
\ea   
which are the nontrivial facets of the Bell polytope. All nonlocal correlations lie outside the Bell polytope and violate a Bell-CHSH inequality.
\section{The two measures and $3$-decomposition of NS boxes}\label{BDaMD}
\subsection{Bell discord}
The observation that each Bell-CHSH inequality is violated to the algebraic maximum  by only one PR-box 
and a nonlocal correlation cannot violate more than a Bell-CHSH inequality suggests the
trade-off between the 
Bell functions,
\ba
\mathcal{B}_{\alpha\beta} &:=&|\braket{A_0B_0}+(-1)^{\beta }\braket{A_0B_1}+(-1)^{\alpha}\braket{A_1B_0}\nonumber \\ 
&&+(-1)^{\alpha \oplus \beta  \oplus 1} \braket{A_1B_1}|. \label{MBF}
\ea
\begin{observation}
For any given nonsignaling box, $P(a_m,b_n|A_i,B_j)$, the Bell functions in Eq. (\ref{MBF}) satisfy the monogamy relationship, 
\be
\mathcal{B}_{00}+\mathcal{B}_{j}\le4, \quad \forall j=01,10,11. \label{BFm}
\ee 
\end{observation}
\begin{proof}
Since $\mathcal{B}_{\alpha\beta}\le2$ for all the local boxes, the trade-off relations in Eq. (\ref{BFm}) are satisfied by any box in the Bell polytope.
It is obvious that all the eight PR-boxes satisfy the trade-off since for any PR-box only one of the Bell functions attains the value $4$ and the rest of them are zero.
Geometrically, any box in the nonlocal region (see fig. \ref{NS3dfig}) lies on a line joining a PR-box and a Bell-local box which lies on the facet 
of the local polytope, i.e.,
any nonlocal box can be decomposed as follows, 
\be
P_{NL}=pP^{\alpha\beta\gamma}_{PR}+(1-p)P_L, \label{GNLd}
\ee
where $P_L$ gives the local bound of a Bell-CHSH inequality.
Now we consider the nonlocal boxes which maximize the
left-hand side of the trade-off in Eq. (\ref{BFm}); for instance, any convex mixture of the PR-box
and the deterministic box, $P=pP^{000}_{PR}+(1-p)P^{0000}_D$,  
gives $\mathcal{B}_{00}+\mathcal{B}_{j}=4$, $\forall j=01,10,11$.
\end{proof}
The Bell function monogamy given in Eq. (\ref{BFm}) refers to the monogamy of a given correlation with respect to the different Bell-CHSH inequalities, 
whereas the conventional monogamy refers to the 
monogamy of a given Bell-type inequality with respect to the different marginal correlations of a given multipartite correlation \cite{Bellmono}. 

We observe that the Bell function monogamy of a PR-box does not vanish for local boxes also if they have an irreducible PR-box component, for instance, 
any isotropic PR-box,
\be
P=pP^{\alpha\beta\gamma}_{PR}+(1-p)P_N, \label{isoPR}
\ee
has a special property that only one of the Bell functions is nonzero which is due to the irreducible PR-box, $P^{\alpha\beta\gamma}_{PR}$, in the decomposition. 
Thus, this property quantifies Bell discord of the isotropic PR-boxes.
%however, there are local correlations that have an irreducible PR-box component which
%can have more than one Bell functions nonzero.
%When $\alpha\beta\gamma=000$ in Eq. (\ref{isoPR}), the box has $\mathcal{B}_{00}=4p$ and the rest of the Bell functions are zero. 

We exploit the Bell function monogamy of the extremal boxes to define the measure of Bell discord which quantifies the irreducible PR-box component in any box. 
Before defining Bell discord we construct the following quantities,
\ba 
\mathcal{G}_1&:=&\Big||\mathcal{B}_{00}-\mathcal{B}_{01}|-|\mathcal{B}_{10}-\mathcal{B}_{11}|\Big|\nonumber\\
\mathcal{G}_2&:=&\Big||\mathcal{B}_{00}-\mathcal{B}_{10}|-|\mathcal{B}_{01}-\mathcal{B}_{11}|\Big| \\
\mathcal{G}_3&:=&\Big||\mathcal{B}_{00}-\mathcal{B}_{11}|-|\mathcal{B}_{01}-\mathcal{B}_{10}|\Big| \nonumber.
\ea
Here $\mathcal{G}_i$ are constructed such that it satisfies the following properties: (i) positivity i.e., $\mathcal{G}_i\ge0$, (ii) $\mathcal{G}_i=0$ for all the
deterministic boxes and (iii) the algebraic maximum of $\mathcal{G}_i$ is achieved by the PR-boxes i.e., $\mathcal{G}_i=4$ for any PR-box. 
%Thus, any nonsignaling box
%takes $0\le\mathcal{G}_i\le4$.
\begin{definition}\label{BDdef}
Bell discord, $\mathcal{G}$, is defined as,
\begin{equation}
\mathcal{G} := \min_i \mathcal{G}_i. \label{defBD}
\end{equation}
Here $0\le\mathcal{G}\le4$. 
\end{definition}
Bell discord is clearly invariant under LRO and 
interchange of the subsystems since the set $\{\mathcal{G}_i, i=1,2,3\}$ is invariant under these two transformations. Therefore,
a $\mathcal{G}>0$ box cannot be transformed into a $\mathcal{G}=0$ box by LRO and vice versa.
Before characterizing the $\mathcal{G}>0$ boxes we make the following two observations.   
\begin{observation}
The set of local boxes that have $\mathcal{G}=0$ forms a subset of the set of all local boxes and is nonconvex.
\end{observation}
\begin{proof}
The set of $\mathcal{G}=0$ boxes is nonconvex since certain convex combination of the $\mathcal{G}=0$ boxes can have $\mathcal{G}>0$; for instance, 
the boxes in Eq. (\ref{isoPR}) can be written as a
convex combination of the deterministic boxes when $p\le\frac{1}{2}$, however, it has Bell discord $\mathcal{G}=4p>0$ if $p>0$.
As the extremal boxes of the Bell polytope have $\mathcal{G}=0$ and the Bell polytope contains $\mathcal{G}>0$ boxes, the set of $\mathcal{G}=0$
boxes form a subset of the local boxes. 
%The set of local boxes that have $\mathcal{G}=0$ forms a subpolytope inside the local polytope in that 
%the extremal boxes of the $\mathcal{G}=0$ region are the deterministic boxes
%and other nonextremal $\mathcal{G}=0$ boxes can be written as a convex combination of the deterministic boxes. 
%However, the $\mathcal{G}=0$ polytope is nonconvex   
\end{proof}
\begin{figure}
\centering
\includegraphics[scale=0.40]{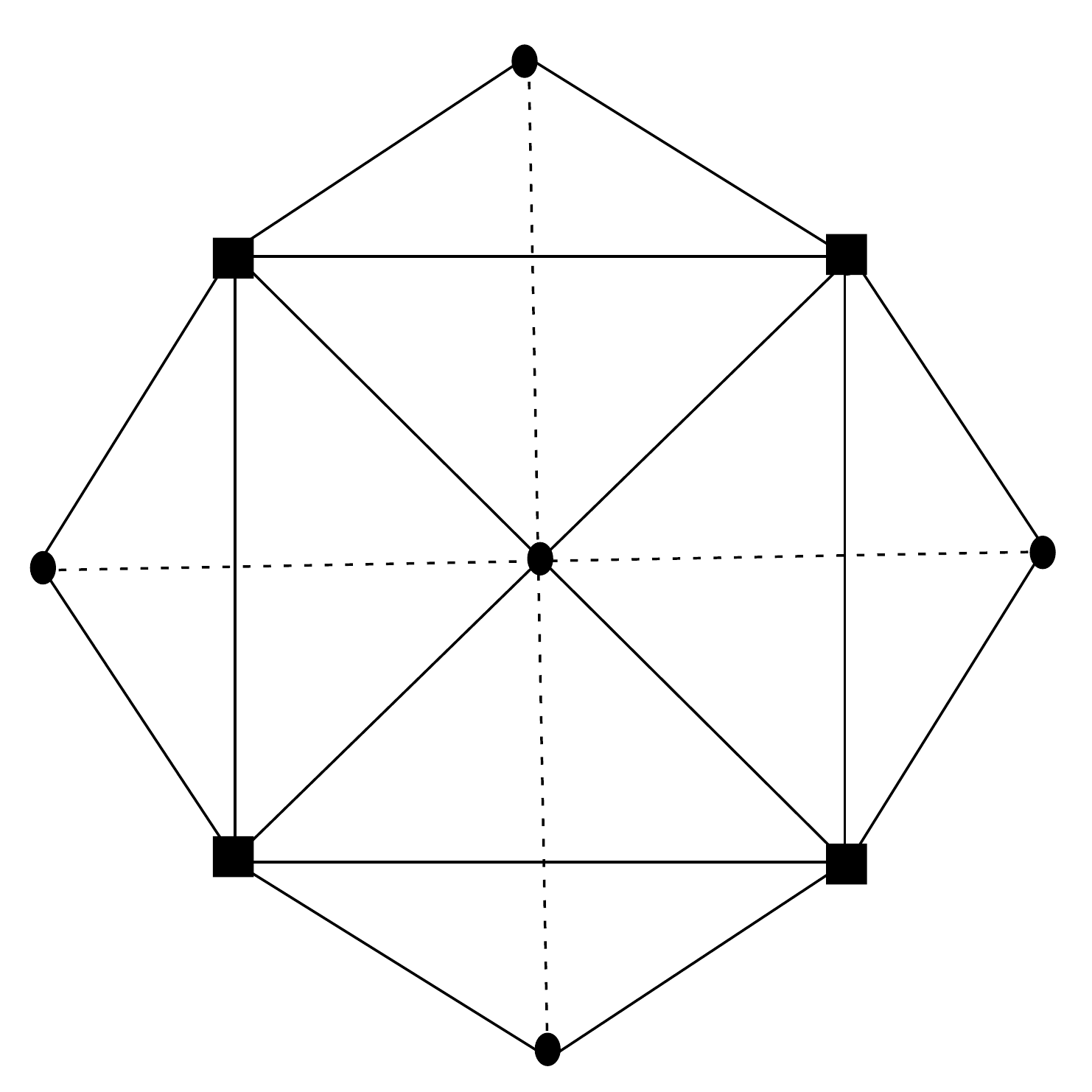} 
\caption{A two-dimensional representation of the NS polytope is shown here. Square represents the local polytope whose
vertices denoted by square points represent the deterministic boxes.
The circular points which lie above the local polytope represent the PR-boxes.
The points which lie on the lines connecting the center of the NS polytope (white noise) and the square points forms $\mathcal{G}=0$ nonconvex polytope.
Any point that goes outside the $\mathcal{G}=0$ region lies on a line joining a PR-box and a $\mathcal{G}=0$ box; for instance, any point that lies
on the dotted line can be written as a convex mixture of a PR-box and white noise.
}\label{Gpolytope}
\end{figure} 

\begin{observation}\label{umPR}
The unequal mixture of any two PR-boxes: $pP^i_{PR}+qP^j_{PR}$, here $p>q$, can be written as the mixture of an irreducible PR-box and a Bell-local box.
\end{observation}
\begin{proof}
$pP^i_{PR}+qP^j_{PR}=(p-q)P^i_{PR}+2qP^{ij}_l$. Here $P^{ij}_l=\frac{1}{2}(P^i_{PR}+P^j_{PR})$ is a Bell-local box since uniform mixture of any two PR-boxes
does not violate a Bell-CHSH inequality.  
%For any PR-box, $\mathcal{B}_{\alpha\beta\gamma} \in \{-4, 0, 4\}$. 
%Further, no two PR-boxes can simultaneously take $4$ or $-4$ for any $\mathcal{B}_{\alpha\beta\gamma}$, which, in turn, implies that $\mathcal{B}_{\alpha\beta\gamma}\in\{-2, 0, 2\}$ for all $P_l^{ij}$. 
Notice that the second PR-box, $P^j_{PR}$,  
in the unequal mixture is not irreducible as it can vanish with the first PR-box in the other possible decomposition by the uniform mixture. 
\end{proof}
\begin{observation}\label{Girre}
$\mathcal{G}$ calculates the irreducible PR-box component in the mixture of the $8$ PR-boxes: $\sum^7_{k=0} p_k P^k_{PR}$ given in Eq. (\ref{CHNS}). 
\end{observation}
\begin{proof}
Notice that $P^{k+1}_{PR}$ is the anti-PR-box to $P^{k}_{PR}$ with $k=0,2,4,6$ since uniform mixture of these two PR-boxes gives white noise \cite{Short}.
The evaluation of $\mathcal{G}_1$ for the mixture of the $8$ PR-boxes gives,  
\be
\mathcal{G}_1\left(\sum_k p_k P^k_{PR}\right)=4|\Big||p_0-p_1|-|p_2-p_3|\Big|-\Big||p_4-p_5|-|p_6-p_7|\Big||.
\ee
The observation \ref{umPR} implies that the terms $|p_k-p_{k+1}|$ in this equation give the irreducible PR-box component 
in the mixture of the two PR-boxes whose equal mixture gives white noise. Thus, 
$\min_i\mathcal{G}_i\left(\sum_k p_k P^k_{PR}\right)$ gives the irreducible PR-box component in the mixture of 
the $4$ reduced components of the PR-boxes that does not contain any anti-PR-box.  
\end{proof}
\begin{observation}\label{nllg0}
Any NS box can be decomposed in a convex mixture of a nonlocal box and a local box with $\mathcal{G}=0$,
\be
P=\eta P_{NL}+ (1-\eta)P_L^{\mathcal{G}=0}. \label{cNlL}
\ee
\end{observation}
\begin{proof}
Since the set of NS boxes is convex and the Bell polytope is contained inside the full NS polytope,
any NS box lies on a line joining a nonlocal box and a local box.
Suppose the local box in the canonical decomposition given in Eq. (\ref{cNlL}) has $\mathcal{G}>0$, 
then it cannot, in general, represent the $\mathcal{G}=0$ boxes. Thus, the division of the Bell polytope into a 
$\mathcal{G}>0$ region and $\mathcal{G}=0$ region allows us to write any NS box as a convex mixture of a nonlocal box and a local box with $\mathcal{G}=0$.
\end{proof}

We obtain the following canonical decomposition of the NS boxes.
\begin{theorem}\label{thm1}
Any NS box can be decomposed into PR-box and a local box that does not have an irreducible PR-box component,
\be
P=\mu P^{\alpha\beta\gamma}_{PR}+\left(1-\mu\right)P_{L}^{\mathcal{G}=0}, \label{Gde}
\ee
where $\mu$ is the irreducible PR-box component and $P_{L}^{\mathcal{G}=0}$ is the local box which has $\mathcal{G}=0$.
\end{theorem}
\begin{proof}
We write any NS box given by the decomposition in Eq. (\ref{CHNS}) as a convex combination of the $8$ PR-boxes and a restricted local box 
that cannot be written as a convex sum of the PR-boxes and the deterministic boxes:
\be
P=\sum^7_{k=0} g_k P^k_{PR} +\left(1-\sum^7_{k=0} g_k\right)P_L; \quad k=\alpha\beta\gamma, \label{step1}
\ee
where $P_L\ne \sum_k r_k P^k_{PR}+\sum_l s_l P^l_D$, i.e., $P_L$ cannot have nonzero $r_k$ overall possible decompositions. 
We are now interested in reducing the combination of the $8$ PR-boxes in Eq. (\ref{step1}) into an irreducible PR-box and a local box
by using the procedure given in observation \ref{umPR}.
It follows from the observation \ref{Girre} that we should first reduce the mixture of the $8$ PR-boxes in Eq. (\ref{step1}) 
into the mixture of the $4$ PR-boxes which does not contain any anti-PR-box, and white noise. Then,
we further reduce it to the mixture of an irreducible PR-box and the local boxes which are the uniform mixture of the two PR-boxes:
\be
\sum^7_{k=0}g_kP^{k}_{PR}=\mu P^{\alpha\beta\gamma}_{PR}+\sum^3_{l=1} p_lP^l_L+p_NP_N. \label{step2}
\ee
Here $\mu$ is obtained by minimizing the PR-box component over all possible decompositions  
i.e., $\mu>0$ iff $\sum^7_{k=0}g_kP^{k}_{PR}\ne\sum^3_{l=1} q_lP^l_L+p_NP_N$ (see Appendix. \ref{irreducible} for illustration). 
Here $P^l_L$ are the maximally-local boxes since the
uniform mixture of any two PR-boxes from the set of the four PR-boxes that does not contain any anti-PR-box
gives the local bound of a Bell-CHSH inequality.
Now substituting Eq. (\ref{step2}) in Eq. (\ref{step1}), 
we get the following decomposition of any NS box,
\be
P=\mathcal{\mu} P^{\alpha\beta\gamma}_{PR}+(1-\mu)P_L.  \label{proofgnz}
\ee
Here $P_L=\frac{1}{1-\mathcal{\mu}}\left\{\sum^3_{l=1} p_lP^l_L+p_NP_N+\left(1-\sum_k g_k\right)P_L\right\}$.
This local cannot have an irreducible PR-box component since $\mu$ is the maximal irreducible PR-box component. 
Further, it follows from the observation \ref{nllg0} that the local box in Eq. (\ref{proofgnz}) must have $\mathcal{G}=0$. 
\end{proof}

We now show that a box has nonzero Bell discord iff it admits a decomposition that has an irreducible PR-box component. 
For any box given by the decomposition in Eq. (\ref{Gde}),  
$\mathcal{G}$ is linear (see Appendix \ref{lbdmd} for illustration) i.e., $\mathcal{G}(P)=\mu\mathcal{G}\left(P^{\alpha\beta\gamma}_{PR}\right)+\left(1-\mu
\right)\mathcal{G}\left(P^{\mathcal{G}=0}_L\right)$ which implies $\mathcal{G}(P)=4\mu>0$ iff $\mu>0$. 
Thus, if a box has nonzero Bell discord, it lies on a line joining a PR-box and a local box that does not have an irreducible PR-box component.  
The violation of a Bell-CHSH inequality is only 
a sufficient condition for nonzero Bell discord which can be illustrated by evaluating the Bell-CHSH operator $\mathcal{B}_{000}$
for the box given by the decomposition in Eq. (\ref{Gde}) with $\alpha\beta\gamma=000$ and $l=\mathcal{B}_{000}\left(P^{\mathcal{G}=0}_L\right)>0$. 
This box violates the Bell-CHSH inequality, i.e., $\mathcal{B}_{000}=4\mu+l(1-\mu)>2$ iff $\mu>\frac{2-l}{4-l}$. 
Suppose $l=0$, the box violates the Bell-CHSH inequality if $\mu>\frac{1}{2}$, whereas it has nonzero Bell discord if $\mu>0$.

We say that the decomposition of the NS boxes given in Eq. (\ref{Gde}) is canonical in that it represents the classification of any NS box
according to whether it has Bell discord or not, which is more general than the classification of NS boxes into nonlocal and local boxes. 
In Sec. \ref{prl}, we observed that local boxes which do not have Bell discord
can have Mermin discord due to the irreducible Mermin box component which in turn implies that 
the canonical description of NS boxes given in Eq. (\ref{Gde}) is not the most general as the $\mathcal{G}=0$ box in this decomposition can have Mermin discord.
%Most general canonical description is thus given by the $3$-decomposition. 

\subsection{Mermin boxes}Before we define a measure of Mermin discord we introduce all the Mermin boxes which lie on the Bell polytope. 
We consider the following maximally-local box,
\begin{equation}
P^{nm}_M = \left( \begin{array}{cccc}
1 & 0 & 0 & 0 \\
\half  & \half  & 0 & 0 \\
\half & 0 & \half  & 0 \\
0 & \half & \half &  0
\end{array} \right).
\label{eq:nmerminbox}
\end{equation}
This box also exhibits the logical contradiction with realistic value assignment shown by the Mermin box in Eq. (\ref{eq:merminbox}). 
Notice that these two Mermin boxes differ by their 
marginals; the one in Eq. (\ref{eq:merminbox}) has maximally mixed marginals, whereas the Mermin box in Eq. (\ref{eq:nmerminbox}) has nonmaximally
mixed marginals.
Thus, the Bell polytope admits two types of Mermin boxes which can be distinguished by their marginals. The following $32$ Mermin boxes:
\ba
P_M^{\alpha\beta\gamma\epsilon}=\frac{1}{2}(\delta^i_{m\oplus
  i\oplus\alpha}\delta^j_{n\oplus   j\oplus\beta}   +\delta^i_{m\oplus
  \gamma}\delta^j_{n\oplus\epsilon}),\nonumber\\
   P_{M'}^{\alpha\beta\gamma\epsilon}=\frac{1}{2}(\delta^i_{m\oplus
  i\oplus\alpha}\delta^j_{n\oplus\beta}+\delta^i_{m\oplus
  \gamma}\delta^j_{n\oplus j \oplus \epsilon}), \label{nMmmm} 
\ea 
which are equal the mixture of two deterministic boxes, can be obtained from the Mermin box in Eq. (\ref{eq:nmerminbox}) by LRO. The following $8$ Mermin boxes:
\begin{align}
P_M^{\alpha\beta\gamma}(a_m,b_n|A_i,B_j)&=\left\{
\begin{array}{lr}
\frac{1}{4}, & i\oplus j =0 \\
\frac{1}{2}, & m\oplus n=i\cdot j \oplus  \alpha  i \oplus  \beta j \oplus \gamma\\ 
0 , & \text{otherwise},\nonumber\\
\end{array}
\right.   
\end{align}
here $\alpha\beta\gamma=00\gamma,10\gamma$, and,
for $\alpha\beta\gamma=01\gamma,11\gamma$,
\begin{align}
P_M^{\alpha\beta\gamma}(a_m,b_n|A_i,B_j)
&=\left\{
\begin{array}{lr}
\frac{1}{4}, & i\oplus j =1 \\
\frac{1}{2}, & m\oplus n=i\cdot j \oplus  \alpha  i\oplus \beta j \oplus  \gamma  \\ 
0 , & \text{otherwise},\\
\end{array} \label{Mmmm}
\right.   
\end{align}
which are the equal mixture of four deterministic boxes, can be obtained from the Mermin box in Eq. (\ref{eq:merminbox}) by LRO. 
As all the Mermin boxes are maximally-local, they lie on the facet of the Bell polytope (see fig. \ref{NS3dfig}).

Similar to PR-boxes which are locally equivalent and maximally nonlocal, 
the Mermin boxes with maximally mixed marginals are locally equivalent and maximally nonclassical. 
The following analogy with PR-box would help us to understand how the Mermin box is nonclassical despite being noncontextual and nonextremal with respect to the NS polytope:
A PR-box is not extremal with respect to the signaling polytope since it can be decomposed into the uniform mixture of two nonlocal deterministic boxes \cite{Bub}; for instance,
the canonical PR-box, $P^{000}_{PR}$, can be written as the uniform mixture of two signaling deterministic boxes which violate the same Bell-CHSH inequality to the algebraic
maximum,
\be
P^{000}_{PR}=\frac{1}{2}\left(\delta_m^0 \delta_n^{i \cdot j}+\delta_m^1 \delta_n^{i \cdot j}\right).
\ee
Thus, signaling is disappeared by the uniform mixture, however, nonlocality of the two signaling boxes does not disappear
as the uniform mixture again maximally violates the Bell-CHSH inequality.
Similarly, the Mermin boxes with maximally mixed marginals admit a decomposition into the uniform mixture of two PR-boxes. 
A uniform mixture of two PR-boxes can also give rise to white noise, but a Mermin box is a special kind of maximally-local box that has nonclassicality 
while nonlocality is disappeared by the uniform mixture. For instance, 
the Mermin box in Eq. (\ref{eq:merminbox}) can be decomposed as follows,
\be
P_M=\frac{1}{2}\left(P^{000}_{PR}+P_{PR}^{110}\right).
\ee
Notice that the correlated box in Eq. (\ref{eq:ccbox}) cannot be decomposed into the uniform mixture of two PR-boxes.

Since all the Mermin boxes have $\mathcal{G}=0$, any $\mathcal{G}=0$ box in the Bell polytope can be written as a convex mixture of the Mermin boxes and 
the deterministic boxes. We will use this observation to obtain the $3$-decomposition fact from the canonical decomposition given in Eq. (\ref{Gde}).   
\subsection{Mermin discord} %Having introduced the Mermin boxes we now introduce a measure for Mermin discord.
%In Sec. \ref{prl}, we noticed that the maximally entangled state exhibits Peres' paradox when it gives rise to Mermin box.
%Mermin argued that GHZ paradox generalizes the Peres's paradox to the multipartite scenario \cite{UNLH}.
%It is known that the correlation which exhibits the GHZ paradox violates 
%a  Mermin inequality maximally \cite{mermin}.   
%The triparite generalization of the Mermin boxes is related to GHZ paradox \cite{UNLH} and a violate a Mermin inequality maximally
We consider the Mermin inequalities:
\ba
\mathcal{M}_{\alpha\beta\gamma}&:=
&(\alpha\oplus\beta\oplus1)\{(-1)^{\beta}\braket{A_0B_1}\!+\!(-1)^{\alpha}\braket{A_1B_0}\}\nonumber\\ 
&&+(\alpha\oplus\beta)\{(-1)^{\gamma}\braket{A_0B_0}+(-1)^{\alpha\oplus\beta\oplus\gamma\oplus 1}\braket{A_1B_1}\}\nonumber\\
  &&\le2 \quad \text{for} \quad \alpha\beta\gamma=00\gamma,01\gamma;\nonumber \\
\mathcal{M}_{\alpha\beta\gamma}&:=&(\alpha\oplus\beta)\{(-1)^{\beta}\braket{A_0B_1}\!+\!(-1)^{\alpha}\braket{A_1B_0}\}\nonumber\\ 
&&+(\alpha\oplus\beta\oplus1)\{(-1)^{\gamma}\braket{A_0B_0}+(-1)^{\alpha\oplus\beta\oplus\gamma\oplus1}\braket{A_1B_1}\}\nonumber\\
  &&\le2 \quad \text{for} \quad \alpha\beta\gamma=10\gamma,11\gamma, \label{bimi}
\ea 
The multipartite generalization of $\mathcal{M}_{\alpha\beta\gamma}$ generate the Mermin inequalities \cite{mermin,WernerWolfmulti}, hence the name.
Just as the complete set of Bell-CHSH inequalities, the set of these inequalities 
is invariant under LRO and thus it forms a complete set \cite{WernerWolf}. 
The complete set of bipartite Mermin inequalities 
do not distinguish between local and nonlocal correlations since the algebraic maximum 
of any Mermin function, $\mathcal{M}_{\alpha\beta\gamma}$, is $2$ which is the same as the bound given in Eq. (\ref{bimi}). 
However, magnitude of the modulus of the Mermin functions, 
$\mathcal{M}_{\alpha\beta}:=|\mathcal{M}_{\alpha\beta\gamma}|$,
serve to construct Mermin discord which distinguishes Mermin boxes from the other boxes. %Here  
%\ba
%\mathcal{M}_{\alpha\beta}&:=&(\alpha\oplus\beta\oplus1)|(-1)^{\beta}\braket{A_0B_1}+(-1)^{\alpha}\braket{A_1B_0}|\nonumber\\ 
%&&+(\alpha\oplus\beta)|\braket{A_0B_0}+(-1)^{\alpha\oplus\beta\oplus1}\braket{A_1B_1}|, \nonumber \\
%&&\text{for} \quad \alpha\beta=00,01; \nonumber\\
%\mathcal{M}_{\alpha\beta}&:=&(\alpha\oplus\beta)|(-1)^{\beta}\braket{A_0B_1}+(-1)^{\alpha}\braket{A_1B_0}|\nonumber\\ 
%&&+(\alpha\oplus\beta\oplus1)|\braket{A_0B_0}+(-1)^{\alpha\oplus\beta\oplus1}\braket{A_1B_1}|\nonumber \\
%&&\text{for} \quad \alpha\beta=10,11.     
%\ea

We observe that for any Mermin box, only one of the Mermin functions, $\mathcal{M}_{\alpha\beta}$, attains $2$ and the rest of them are zero, 
whereas for the deterministic boxes and the PR-boxes, two of the Mermin functions attains $2$ and the other two are zero.
We exploit this property of the extremal boxes with respect to the complete set of Mermin functions $\{\mathcal{M}_{\alpha\beta}\}$ 
to define Mermin discord.
\begin{definition}\label{defMD}
Mermin discord, $\mathcal{Q}$, is defined as,
\begin{equation}
\mathcal{Q} := \min_j \mathcal{Q}_j,
\end{equation}    
where, $\mathcal{Q}_1=\Big||\mathcal{M}_{00}-\mathcal{M}_{01}|-|\mathcal{M}_{10}-\mathcal{M}_{11}|\Big|$, and $\mathcal{Q}_2$ and $\mathcal{Q}_3$ 
are obtained by permuting $\mathcal{M}_{\alpha\beta}$ in $\mathcal{Q}_1$. Here $0\le\mathcal{Q}\le2$. 
\end{definition}
Mermin discord is constructed such that all the PR-boxes and the deterministic boxes have $\mathcal{Q}=0$, and, 
the algebraic maximum of $\mathcal{Q}$ is achieved by the Mermin boxes 
i.e., $\mathcal{Q}=2$ for any Mermin box. Mermin discord is clearly invariant
under LRO and permutation of the parties as the set $\{\mathcal{Q}_j\}$ is invariant under these two transformations. 

We obtain the following observations from the Mermin discord defined in \ref{defMD}.
\begin{observation} \label{qzp}
The set of $\mathcal{Q}=0$ boxes forms a nonconvex subset of the set of all NS boxes. 
\end{observation}
\begin{proof}
The set of $\mathcal{Q}=0$ boxes is nonconvex since certain convex mixture of the $\mathcal{Q}=0$ boxes can have $\mathcal{Q}>0$. 
Since the PR-boxes and the deterministic boxes have $\mathcal{Q}=0$, 
the set of $\mathcal{Q}=0$ boxes forms 
a nonconvex region in the full NS polytope.
\end{proof}

\begin{observation}\label{qdg}
$\mathcal{Q}$ divides the $\mathcal{G}=0$ region into a $\mathcal{Q}>0$ region and $\mathcal{G}=\mathcal{Q}=0$ nonconvex region. 
\end{observation}
\begin{proof}
Since all the deterministic boxes have $\mathcal{G}=\mathcal{Q}=0$
and the certain convex combination of the deterministic boxes can give rise to a $\mathcal{Q}>0$ box,
the set of $\mathcal{G}=\mathcal{Q}=0$ boxes forms a nonconvex subregion of the $\mathcal{G}=0$ region.
\end{proof}

\begin{observation}\label{mlq2}
A maximally-local box that has $\mathcal{Q}=2$ is, in general, a convex combination of a maximally mixed marginals Mermin box and the four 
nonmaximally mixed marginals Mermin boxes
which are equivalent with respect to $\braket{A_iB_j}$,
\be
P^{\alpha\beta\gamma}_{\mathcal{Q}=2}=\sum^4_{i=1}p_{M_i}P^{nm}_{M_i}+p_MP^{\alpha\beta\gamma}_M,
\ee
where $P^{nm}_{M_i}$ are the four nonmaximally mixed marginals Mermin boxes which all have the same values for $\braket{A_iB_j}$ and 
$P^{\alpha\beta\gamma}_M=\frac{1}{4}\sum^4_{i=1}P^{nm}_{M_i}$ is one of the eight maximally mixed marginals Mermin boxes.  
\end{observation}
\begin{proof}
Notice that the two Mermin boxes in Eqs. (\ref{eq:merminbox}) and (\ref{eq:nmerminbox}) have the the same property with respect to $\braket{A_iB_j}$ i.e., 
they have $\braket{A_0B_0}=-\braket{A_1B_1}=1$
and $\braket{A_0B_1}=\braket{A_1B_0}=0$ which implies that any convex mixture of these two Mermin boxes again have $\mathcal{Q}=2$. There are 
four nonmaximally mixed marginals Mermin boxes which are equivalent with respect to $\braket{A_iB_j}$ corresponding to a given maximally mixed marginals Mermin box. 
Thus, any convex mixture of these five Mermin boxes is again a $\mathcal{Q}=2$ box. It can be checked that equal mixture of the four nonmaximally mixed marginals Mermin boxes 
which are equivalent with respect to $\braket{A_iB_j}$ gives the maximally mixed marginals Mermin box.
\end{proof}

\begin{observation}\label{imbc}
The unequal mixture of any two Mermin boxes which differ by $\braket{A_iB_j}$: $pP^1_M+qP^2_M$; $p>q$, 
can be written as a convex mixture of an irreducible Mermin box and a $\mathcal{Q}=0$ box.  
\end{observation}
\begin{proof}
$pP^1_M+qP^2_M=(p-q)P^1_M+2qP_{\mathcal{Q}=0}$. Here $P_{\mathcal{Q}=0}=\frac{1}{2}(P^1_M+P^2_M)$
is a $\mathcal{Q}=0$ box since it is the uniform mixture of the two Mermin boxes which differ by $\braket{A_iB_j}$.  
\end{proof}

We now obtain the following $3$-decomposition from these observations and the result \ref{thm1}. 
\begin{theorem}
Any NS box can be written as a convex mixture of a PR-box, a maximally-local box with $\mathcal{Q}=2$ and a local box with $\mathcal{G}=\mathcal{Q}=0$,
\be
P=\mu P^{\alpha\beta\gamma}_{PR}+\nu P^{\alpha\beta\gamma}_{\mathcal{Q}=2}+(1-\mu-\nu)P^{\mathcal{G}=0}_{\mathcal{Q}=0}. \label{pDecomp}
\ee
\end{theorem}
\begin{proof}
The $\mathcal{G}=0$ box in the canonical decomposition in Eq. (\ref{proofgnz}) is given by,
\be
P^{\mathcal{G}=0}_L=\frac{1}{1-\mathcal{\mu}}\left\{\sum^3_{l=1} p_lP^l_L+p_NP_N+\left(1-\sum_k g_k\right)P_L\right\}. \label{cgz}
\ee
Notice that the first term can have maximally mixed marginals Mermin boxes since it has maximally-local boxes which are uniform mixture of two PR-boxes
and the last term can have nonmaximally mixed marginals Mermin boxes. 
Suppose the $\mathcal{G}=0$ box in Eq. (\ref{cgz}) has Mermin box components, it follows from the observations (\ref{qdg})-(\ref{imbc})  
that it can be decomposed into
an irreducible maximally-local box that has $\mathcal{Q}=2$ and a $\mathcal{G}=\mathcal{Q}=0$ box,
\be
P^{\mathcal{G}=0}_L= q_MP^{\alpha\beta\gamma}_{\mathcal{Q}=2}+(1-q_M)P^{\mathcal{G}=0}_{\mathcal{Q}=0}. \label{GGO}
\ee
Here $q_M$ is the maximal irreducible component of the  box with $\mathcal{Q}=2$
and $P^{\mathcal{G}=0}_{\mathcal{Q}=0}$ is a local box with $\mathcal{G}=\mathcal{Q}=0$.
Substituting Eq. (\ref{GGO}) in Eq. (\ref{Gde}), we get the following canonical decomposition for any NS box,
\be
P=\mu P^{\alpha\beta\gamma}_{PR}+{\nu}P^{\alpha\beta\gamma}_{\mathcal{Q}=2}+(1-{\mu}-{\nu})
P^{\mathcal{G}=0}_{\mathcal{Q}=0}, 
\ee
where $\nu=\left(1-\mu\right)q_M$. In this decomposition, the box in the third term must have $\mathcal{G}=\mathcal{Q}=0$. Otherwise,
it would lead to the contradiction that if it had $\mathcal{Q}>0$ or $\mathcal{G}>0$, then the  decomposition would not
represent the boxes in the $\mathcal{G}=\mathcal{Q}=0$ region.
\end{proof}

The $3$-decomposition fact given in Eq. (\ref{pDecomp}) implies that any NS box can also be written in the following canonical form,
\be
P=\nu P^{\alpha\beta\gamma}_{\mathcal{Q}=2}+(1-\nu)P^{\mathcal{G}\ge0}_{\mathcal{Q}=0}, \label{Qcanonical}
\ee
where $P^{\mathcal{G}\ge0}_{\mathcal{Q}=0}=\frac{1}{1-\nu}\left\{\mu P^{\alpha\beta\gamma}_{PR}+(1-\nu-\mu)
P^{\mathcal{G}=0}_{\mathcal{Q}=0}\right\}$ is a $\mathcal{Q}=0$ box since $P^{\mathcal{G}=0}_{\mathcal{Q}=0}$ cannot have an irreducible 
PR-box component. From the linearity of $\mathcal{Q}$ with respect to the decomposition given in Eq. (\ref{Qcanonical}), it follows that 
$\mathcal{Q}(P)=\nu\mathcal{Q}\left(P^{\alpha\beta\gamma}_{\mathcal{Q}=2}\right)+(1-\nu)\mathcal{Q}\left(P_{\mathcal{Q}=0}\right)=2\nu$.  
Thus, any NS box that has $\mathcal{Q}>0$ lies on a line segment joining a maximally-local box that has $\mathcal{Q}=2$ 
and a $\mathcal{Q}=0$ box.  

\subsection{Monogamy between the measures} 
The probability constraint  $\mu+\nu\le1$ in the $3$-decomposition of any NS box given in Eq. (\ref{pDecomp}) 
implies a trade-off between Bell and Mermin discord as follows,
\be
\mathcal{G}+2\mathcal{Q}\le4.\label{monogamyGQ}
\ee
This trade-off relation reveals monogamy between the two types of nonclassicality of quantum correlations 
which we illustrate with the observations given for quantum correlations in Sec. \ref{prl}.
For the incompatible measurements that gives rise to EPR-steering, the maximally entangled state 
does not give rise to Bell discord since $\mathcal{Q}=2$ implies $\mathcal{G}=0$. 
Due to this monogamous property,
the correlation arising from the nonmaximally entangled states in Eq. (\ref{nmE}) has only Mermin discord i.e., $\mathcal{Q}=2\sin2\theta$ and $\mathcal{G}=0$
for the incompatible measurements that gives rise to the KS paradox.
Similarly, we observe that the correlations arising from the nonmaximally entangled states, $\ket{\psi(\theta)}$, gives rise to only
Bell discord i.e., $\mathcal{G}=2\sqrt{2}\sin2\theta$ and $\mathcal{Q}=0$
for the measurements that gives rise to the Tsirelson box. 
For general incompatible measurements, quantum correlations can have Bell discord and Mermin discord simultaneously,
however, a trade-off exists between Bell and Mermin discord as given in Eq. (\ref{monogamyGQ}); for instance,
the correlations arising from the Bell state in Eq. (\ref{MeCNL}) have Bell discord and Mermin discord simultaneously without violating the trade-off relation between them 
since  these correlations have $\mathcal{G}+\mathcal{Q}=4\sqrt{p}\le4$.
%Thus, when the incompatible measurements performed on a two-qubit state gives rise to quantum correlations can have
%Bell and Mermin discord simultaneously, however, trade-off exists between Bell and Mermin discord as given in Eq. (\ref{monogamyGQ});
%for instance, 
The monogamy relation obtained here is analogous to the monogamy relation between local contextuality and nonlocality 
derived by Kurzy\'{n}ski \etal \cite{KCK} under the no-disturbance principle as both the relations reveal monogamy between two different types of nonclassicality.
\section{Quantum correlations}\label{QC}
We will analyze nonclassicality of quantum correlations arising from local projective measurements on the two-qubit systems along the 
directions $\hat{a}_i$ and $\hat{b}_j$ on the respective qubits.  
%which is equivalent 
%to calculating the factors $\mathcal{G}'$ and $\mathcal{Q}'$, which are invariant under local unitary operations, in the $3$-decomposition.  
The joint state of two-qubit system is given 
by the density operator $\rho_{AB}$ in the complex vector space 
${\cal L}({{\cal H}_A^2\otimes{\cal H}_B^2})$, which in the Bloch representation can be expressed as follows:
\be
\rho_{AB}\!=\!\frac{1}{4}\left[\openone\otimes\openone\!+\!\sum_i r_i\sigma_j\otimes\openone
\!+\!\sum_i s_i\openone\otimes \sigma_i\!+\!\sum_{ij}T_{ij}\sigma_i\otimes\sigma_j\right],
\ee
where $\vec{\sigma}_i$, $i\in{1,2,3}$, are the Pauli matrices, where
$r_i=\mathrm{Tr}(\rho_{AB}\sigma_i\otimes\openone)$ and $s_i=\mathrm{Tr}(\rho_{AB}\openone\otimes\sigma_i)$ are the local Bloch vectors,
and $T_{ij}=\mathrm{Tr}(\rho_{AB}\sigma_i\otimes\sigma_j)$ is the correlation tensor. The set of separable states forms a convex subset of the set of all states.
It is known that the set of zero quantum discord states is nonconvex \cite{QCall,Caves}.
Similarly, as we show, the set of $\mathcal{G}=\mathcal{Q}=0$ boxes is nonconvex. Indeed, we show that any nonzero
quantum discord state which is neither a classical-quantum state nor a quantum-classical state 
can give rise to nonzero Bell discord or/and Mermin discord for suitable incompatible measurements.

%Bell discord and Mermin discord quantify nonlocality and contextuality of the  
%quantum correlations with respect to the $3$-decomposition even in the local correlations as well. 
We will apply Bell discord and Mermin discord to quantify nonclassicality of correlations arising from the pure entangled states and the Werner states.
For these states, a nonzero Bell discord originates from incompatible measurements which give rise to Bell nonlocality.
Similarly, a nonzero Mermin discord originates from incompatible measurements which give rise to EPR-steering.
We will apply these measures to various states in order to illustrate the new insights that may be obtained regarding the origin of
nonclassicality.
\subsection{Maximally entangled state}
We study $3$-decomposition of the correlations arising from the maximally entangled state.
For the measurement settings: 
${\vec{a}_0}=\hat{x}$, ${\vec{a}_1}=\hat{y}$,
${\vec{b}_0} =\sqrt{p}\hat{x}-\sqrt{1-p}\hat{y}$ and ${\vec{b}_1}=\sqrt{1-p}\hat{x}+\sqrt{p}\hat{y}$, where $\frac{1}{2}\le p \le1$, the correlations
arising from the Bell state, $\ket{\psi^+}$, can be decomposed into PR-box, a Mermin box which is an uniform mixture of two PR-boxes and white noise as follows,
\be
P=\mu P^{000}_{PR}+\nu \left(\frac{P^{000}_{PR}+P^{110}_{PR}}{2}\right)+(1-\mu-\nu)P_N, \label{meb1}
\ee 
where $\mu=\sqrt{1-p}$ 
and $\nu=\sqrt{p}-\sqrt{1-p}$. 
These correlations have Bell discord and Mermin discord simultaneously when $\frac{1}{2}<p<1$:
$\mathcal{G}=4\sqrt{1-p}>0$ if $p\ne1$ and $\mathcal{Q}=2(\sqrt{p}-\sqrt{1-p})>0$ if $p\ne\frac{1}{2}$. 
When the Bell state gives rise to a $3$-decomposition, 
the correlation nonmaximally violates a Bell-CHSH inequality and an EPR-steering inequality simultaneously: 
The correlations in Eq. (\ref{meb1}) violate the Bell-CHSH inequality i.e., $\mathcal{B}_{000}=2\left(\sqrt{p}+\sqrt{1-p}\right)>2$ if
$p\ne1$ and the EPR-steering inequality in Eq. (\ref{eprst}) if $p\ne\frac{1}{2}$. 
%Thus, the $3$-decomposition of the Bell state isolates the contextuality of the Bell state from nonlocality. 
When the settings becomes optimal for the violation of the Bell-CHSH inequality which happens at $p=\frac{1}{2}$, 
the correlation has $\mathcal{G}=2\sqrt{2}$ and $\mathcal{Q}=0$. 
When the settings becomes optimal for the violation of the EPR-steering inequality which happens at $p=1$, 
the correlation has $\mathcal{Q}=2$ and $\mathcal{G}=0$. 
%Thus, the three types of correlations which have (i) $\mathcal{G}>0$ and $\mathcal{Q}=0$, (ii) $\mathcal{G}=0$ and $\mathcal{Q}>0$,
%and (iii) $\mathcal{G}>0$ and $\mathcal{Q}>0$ can be obtained by three types of measurement settings.
\subsection{Pure nonmaximally entangled states}
We study the correlations arising from the Schmidt states:
\be
\rho_{S}\!=\!\frac{1}{4}\!\left(\!\openone\!\otimes\!\openone\!+\!c(\!\sigma_z\!\otimes\!\openone\!
+\!\openone\!\otimes\!\sigma_z\!)\!+\!s(\!\sigma_x\!\otimes\!\sigma_x\!-\!\sigma_y\!\otimes\!\sigma_y\!)
\!+\!\sigma_z\!\otimes\!\sigma_z\!\right)\!, \label{Schmidt}
\ee
where $c=\cos2\theta$, $s=\sin2\theta$ and $0\le\theta\le\frac{\pi}{4}$. 
These pure states have Schmidt decomposition \cite{Sch1,Sch2} as given in Eq. (\ref{nmE}), hence the name.
Entanglement of the Schmidt states can be quantified by the tangle, $\tau=s^2$ \cite{tangle}.  
The nonmaximally entangled Schmidt states can give rise to (i) a maximally mixed marginals box when measurements performed lie in the $xy$-plane
or (ii) a nonmaximally mixed marginals box when measurements performed lie in the $xz$-plane.
%When the Schmidt states give rise to maximal Bell discord, the correlations have maximally mixed marginals and 
%do not give rise to optimal violation of the Bell-CHSH inequality.
 
\subsubsection{Bell-Schmidt box}
The Bell-Schmidt boxes are the correlations arising from the Schmidt states which have only nonzero Bell discord. 

(i) For the measurement settings: 
${\vec{a}_0}=\hat{x}$, ${\vec{a}_1}=\hat{y}$,
${\vec{b}_0} =\frac{1}{\sqrt{2}}(\hat{x}-\hat{y})$ and ${\vec{b}_1}=\frac{1}{\sqrt{2}}(\hat{x}+\hat{y})$, 
the Schmidt states give to   
the noisy PR-box which is a mixture of a PR-box and white noise as follows:
\ba
P=\frac{s}{\sqrt{2}}P_{PR}+\left(1-\frac{s}{\sqrt{2}}\right) P_N, \label{BSb}
\ea
These correlations violate the Bell-CHSH inequality i.e., $\mathcal{B}_{000}=2\sqrt{2\tau}>2$ if $\tau>\frac{1}{2}$ 
and have Bell discord $\mathcal{G}=2\sqrt{2\tau}>0$ if $\tau>0$. Notice that the irreducible PR-box component in the local correlations in Eq. (\ref{BSb}) is due 
to the incompatible measurements that gives rise to nonlocality. Thus, Bell discord of these local correlations  
reveals nonclassicality of the nonmaximally entangled states originating from nonlocality. 
%since the correlations have $\mathcal{B}_{00}=2\sqrt{2}s$ and the rest of the Bell functions are zero.   

\begin{figure}
\includegraphics[width=0.45\textwidth]{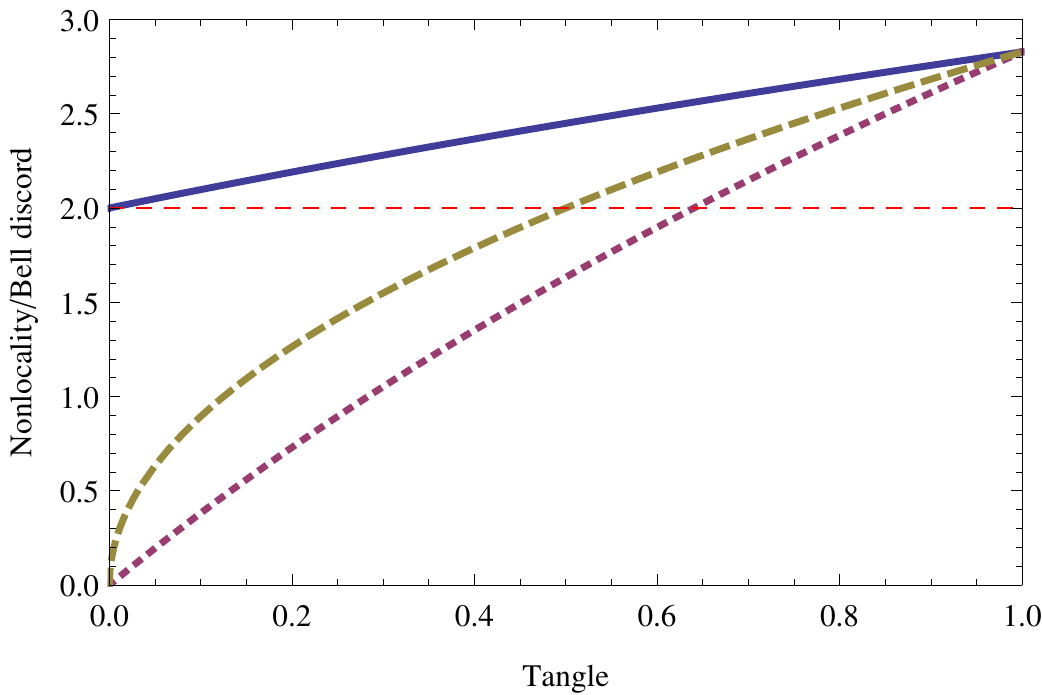} 
\caption{Dashed line shows the plots of the Bell-CHSH inequality violation
and Bell discord for the JPD given in Eq. (\ref{BSb}). Solid and dotted lines show the plots of the Bell-CHSH inequality violation and 
Bell discord respectively for the JPD given in Eq. (\ref{PRQ}).
We observe that the JPD in Eq. (\ref{PRQ}) gives optimal violation of the Bell-CHSH inequality, however, it does not give optimal Bell discord as the JPD
in Eq. (\ref{BSb}) has more Bell discord than this JPD.}\label{plotine}
\end{figure}

(ii) Popescu and Rohrlich showed that all the pure entangled states in Eq. (\ref{Schmidt}) give rise to Bell nonlocality
for the settings \cite{PRQB}: 
${\vec{a}_0}=\hat{z}$, ${\vec{a}_1}=\hat{x}$,
${\vec{b}_0}=\cos t\hat{z}+\sin t\hat{x}$ and ${\vec{b}_1}=\cos t\hat{z}-\sin t\hat{x}$, 
where $\cos t=\frac{1}{\sqrt{1+s^2}}$ as the correlations violate the Bell-CHSH inequality 
i.e., $\mathcal{B}_{000}=2\sqrt{1+\tau}>2$ if $\tau>0$ for this settings. 
These correlations can be decomposed into PR-box and a local box which has $\mathcal{G}=0$ and nonmaximally mixed marginals,
\ba
P&=&s^2\left[\frac{1}{\sqrt{1+s^2}}P_{PR}+\left(1-\frac{1}{\sqrt{1+s^2}}\right)P_N\right]\nonumber\\
&&+\left(1-s^2\right)P^{\mathcal{G}=0}_L(\rho).\label{PRQ}
\ea
Here $P^{\mathcal{G}=0}_L(\rho)$ is a distribution arising from the product state \footnote{$P^{\mathcal{G}=0}_L(\rho)$ is not a valid JPD as the state $\rho$
is not a physical state.}, 
\be
\rho=\rho_A \otimes \rho_B, \label{ScPr}
\ee
where
\be
\rho_A=\rho_B=\frac{1}{2}\left[1+\frac{c}{1-s^2}\right]\ket{0}\bra{0}+\frac{1}{2}\left[1-\frac{c}{1-s^2}\right]\ketbra{1}{1},\nonumber
\ee
for the chosen measurements. %Here the irreducible PR-box always implies Bell nonlocality because of the nonmaximally mixed marginals.
The correlations in Eq. (\ref{PRQ}) have Bell discord $\mathcal{G}=\frac{4\tau}{\sqrt{1+\tau}}>0$ if $\tau>0$.

Notice that the correlations in Eq. (\ref{PRQ}) have less irreducible PR-box component than the correlations in Eq. (\ref{BSb}) for a given amount of entanglement
quantified by the tangle (see fig. \ref{plotine}). Thus,  
when the pure nonmaximally entangled states give rise to optimal violation of the Bell-CHSH inequality, the correlations do not have
optimal Bell discord and has nonmaximally mixed marginals.  
\subsubsection{Mermin-Schmidt box}
Mermin-Schmidt boxes are the local correlations arising from the Schmidt states which have only Mermin discord.
The local correlations which violate an EPR-steering inequality are the subset of the local correlations which have Mermin discord. 
The following Mermin inequalities,
\be
\mathcal{M}_{\alpha\beta\gamma}\le\sqrt{2},
\ee
where $\mathcal{M}_{\alpha\beta\gamma}$ are the Mermin operators given in Eq. (\ref{bimi}), forms the complete set of EPR-steering inequalities 
if the measurement operators on Alice's or Bob's side are anti-commuting qubit observables \cite{EPRsi}; suppose $B_0=\sigma_x$
and $B_1=\sigma_y$, then these inequalities can be obtained from 
the EPR-steering inequality in Eq. (\ref{eprst}) by LRO.

\begin{figure}
\includegraphics[width=0.45\textwidth]{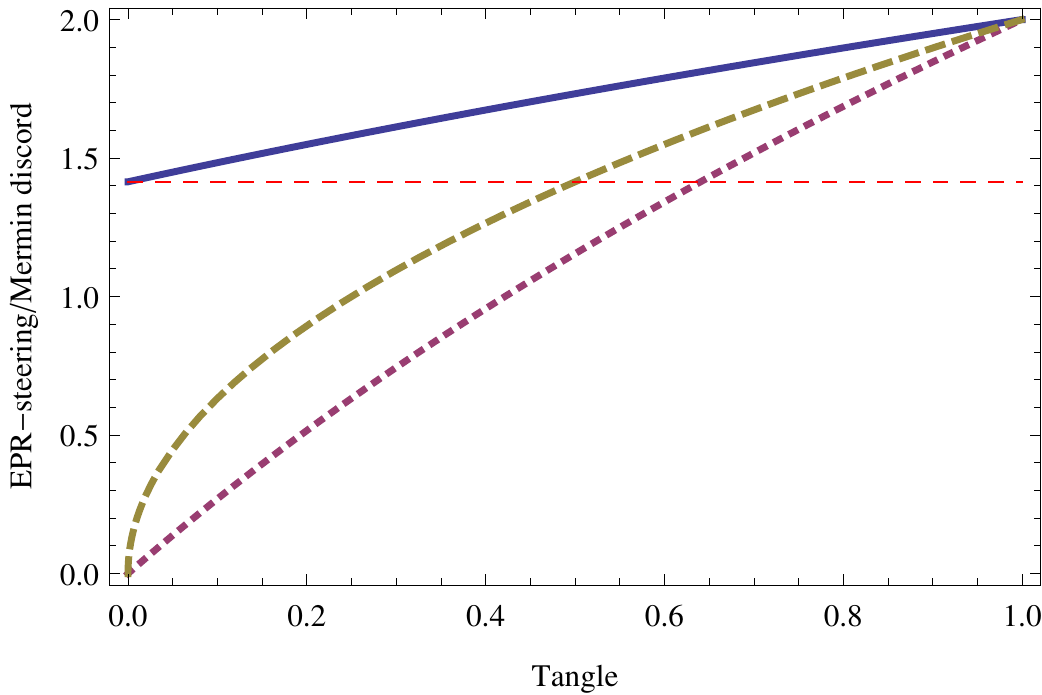} 
\caption{Dashed line shows the plots of the EPR-steering violation
and Mermin discord for the JPD given in Eq. (\ref{MSb}). Solid and dotted lines show the plots of the EPR-steering violation and 
Mermin discord respectively for the JPD given in Eq. (\ref{CSB}). We observe that the JPD in Eq. (\ref{CSB}) has less Mermin discord than
the JPD in Eq. (\ref{MSb}) despite the fact that the former gives rise to optimal violation of the EPR-steering inequality.}\label{figplot2}
\end{figure}

(i) For the settings
${\vec{a}_0}=\hat{x}$, ${\vec{a}_1}=-\hat{y}$,
${\vec{b}_0}=\hat{y}$ and ${\vec{b}_1}=\hat{x}$, the Schmidt states give rise to the noisy Mermin-box which is a mixture of a Mermin box and white noise 
as follows:
\be
P=s\left(\frac{P^{000}_{PR}+P^{111}_{PR}}{2}\right) +(1-s) P_N. \label{MSb}
\ee 
%This box has $\mathcal{M}_{00}=2s$, and,
%$\mathcal{M}_{01}=\mathcal{M}_{10}=\mathcal{M}_{11}=0$, which implies that $\mathcal{Q}=2s$. 
These correlations violate 
the EPR-steering inequality i.e., $\mathcal{M}_{000}=2\sqrt{\tau}>\sqrt{2}$ if $\tau>\frac{1}{2}$ and have Mermin discord $\mathcal{Q}=2\sqrt{\tau}>0$ if $\tau>0$.
Notice that the irreducible Mermin box component in the local correlations in Eq. (\ref{MSb}) is due to the incompatible measurements that gives rise to EPR-steering.
Thus, Mermin discord of these local correlations reveals nonclassicality of the entangled states originating from EPR-steering

 (ii) All the pure entangled states violate the EPR-steering inequality for the settings,
${\vec{a}_0}=\frac{1}{\sqrt{2}}(\hat{z}+\hat{x})$, ${\vec{a}_1}=\frac{1}{\sqrt{2}}(\hat{z}-\hat{x})$,
${\vec{b}_0}=\cos t\hat{z}-\sin t\hat{x}$, and ${\vec{b}_1}=\cos t\hat{z}+\sin t\hat{x}$, 
where $\cos t=\frac{1}{\sqrt{1+s^2} }$ since the Schmidt states give rise to
$\mathcal{M}_{000}=\sqrt{2}\sqrt{1+\tau}>\sqrt{2}$ if $\tau>0$ for this settings.
%, $\mathcal{M}_{01}=\frac{\sqrt{2}c^2}{\sqrt{1+s^2}}$ and 
%$\mathcal{M}_{10}=\mathcal{M}_{11}=0$,
The correlations can be decomposed into Mermin box
and a local box which has $\mathcal{Q}=0$ and nonmaximally mixed marginals,
\ba
P&=&s^2\left[\frac{\sqrt{2}}{\sqrt{1+s^2}}\left(\frac{P^{000}_{PR}+P^{111}_{PR}}{2}\right)+\left(1-\frac{\sqrt{2}}{\sqrt{1+s^2}}\right)P_N\right]\nonumber\\
&&+\left(1-s^2\right)P_{\mathcal{Q}=0}(\rho), \label{CSB}
\ea
where $P_{\mathcal{Q}=0}(\rho)$ is a distribution arising from the state in Eq. (\ref{ScPr}).
The correlations in Eq. (\ref{CSB}) have Mermin discord $\mathcal{Q}=\frac{2\sqrt{2}\tau}{\sqrt{1+\tau}}>0$ if $\tau>0$. 

Similar to the Bell-Schmidt boxes in Eqs. (\ref{BSb}) and (\ref{PRQ}), the correlations in Eq. (\ref{MSb}) have more irreducible Mermin box component than
the correlations in Eq. (\ref{CSB}) for a given amount of entanglement (see fig. \ref{figplot2}).
Thus, when the Schmidt states give rise to optimal violation of an EPR-steering inequality, 
the correlations do not have optimal Mermin discord and have nonmaximally
mixed marginals.   
\subsubsection{Bell-Mermin-Schmidt box}
Bell-Mermin-Schmidt boxes are the correlations arising from the Schmidt states which have nonzero Bell and 
Mermin discord simultaneously.

\begin{figure}
\includegraphics[width=0.45\textwidth]{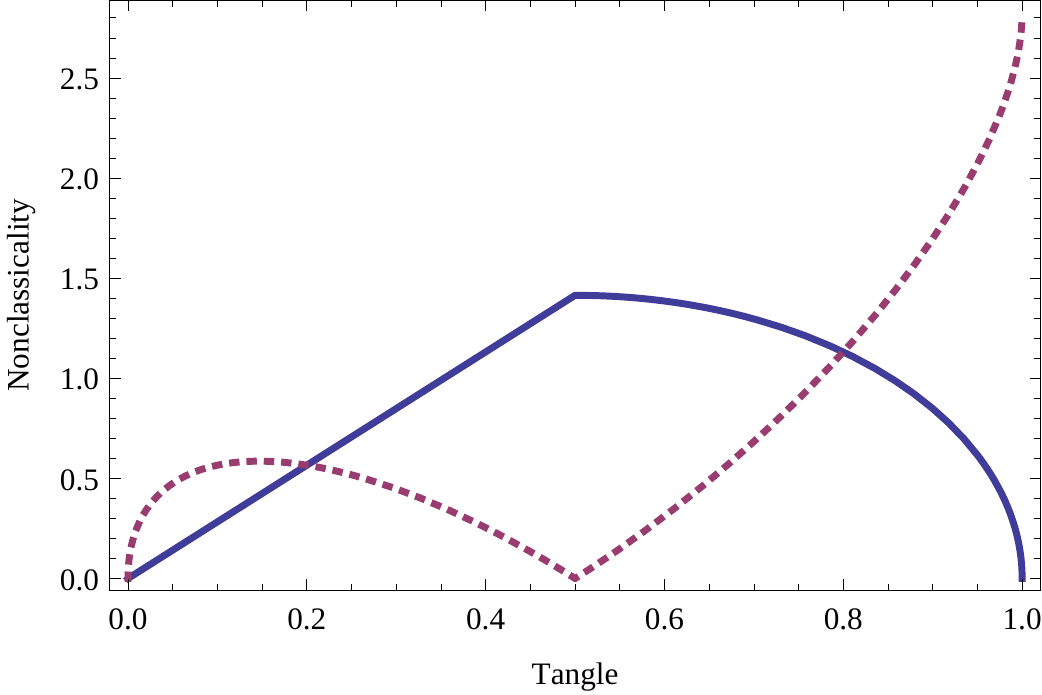} 
\caption{Bell and Mermin discord of the JPD given in Eq. (\ref{BMSb}) are shown by dotted and solid lines respectively.}\label{pt3}
\end{figure}

(i) We define the settings:
${\vec{a}_0}=s\hat{x}+c\hat{y}$,
${\vec{a}_1}=c\hat{x}-s\hat{y}$,
${\vec{b}_0}=\frac{1}{\sqrt{2}}(\hat{x}+\hat{y})$ and
${\vec{b}_1}=\frac{1}{\sqrt{2}}(\hat{x}-\hat{y})$.
For this settings, the correlations have a $3$-decomposition as follows,
\ba
P&=&\left(1-|q|-|g|\right) P_N+ |q|\left(\frac{P^{000}_{PR}+P^{11\gamma}_{PR}}{2}\right)\nonumber\\
&+&|g|\left[\frac{1}{\sqrt{2}}P^{000}_{PR}+\left(1-\frac{1}{\sqrt{2}}\right)P_N\right], \label{BMSb}
\ea
where
$q=\frac{|c-s|-|c+s|}{\sqrt{2}}$ and $g=s(s-c)$.
These correlations have nonzero Bell and Mermin discord as follows (see fig. \ref{pt3}), 
\be
\mathcal{G}=2\sqrt{2\tau}|s-c|>0 \quad \text{except when} \quad \theta  \ne0, \frac{\pi}{8} \nonumber
\ee
and 
\ba
\mathcal{Q}&=&\sqrt{2\tau}\Big||c+s|-|c-s|\Big|>0 \quad \text{except when} \quad \theta \ne0, \frac{\pi}{4}\nonumber\\
&=&\left\{\begin{array}{lr}
2\sqrt{2}\tau \quad \text{when} \quad c>s\\ 
2\sqrt{2\tau(1-\tau^2)} \quad \text{when} \quad s>c.\\ 
\end{array}
\right.\nonumber
\ea
Notice that the correlation in Eq. (\ref{BMSb}) becomes Bell-Schmidt box in Eq. (\ref{BSb}) for $\theta=\pi/4$ 
since the settings becomes optimal for Bell discord 
and becomes Mermin-Schmidt box in Eq. (\ref{MSb}) for $\theta=\pi/8$ since the settings becomes optimal for Mermin discord.

\begin{figure}
\includegraphics[width=0.45\textwidth]{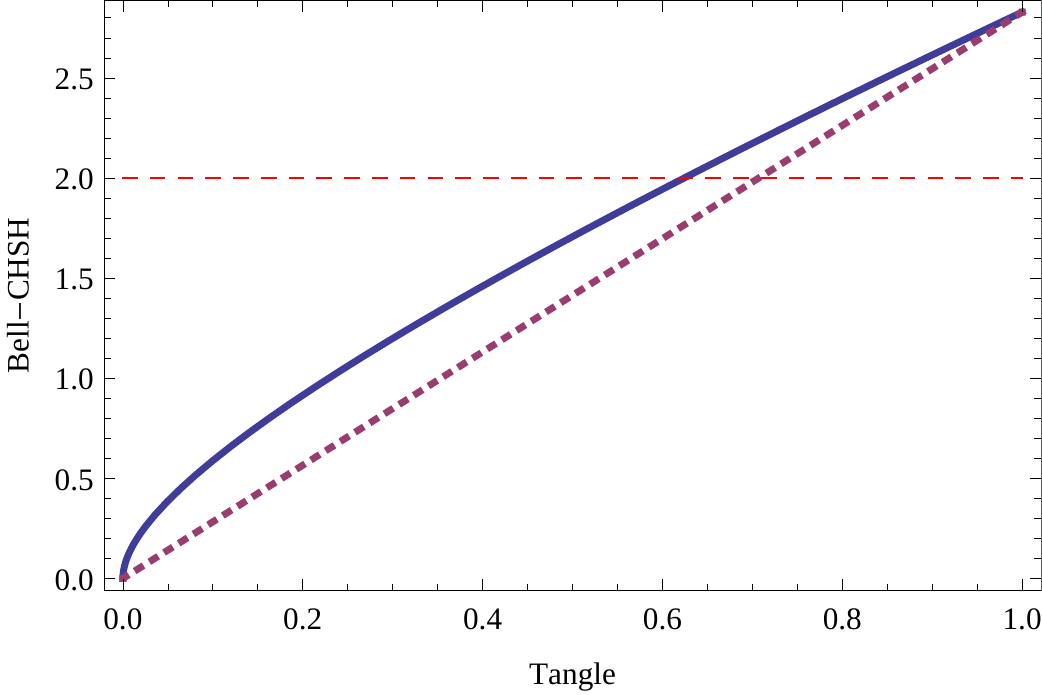} 
\caption{The violation of the Bell-CHSH inequality for the JPD in Eqs. (\ref{BMSb}) and (\ref{BMSb1}) are shown by dotted and solid lines respectively.}\label{pt4}
\end{figure}

(ii) For the settings that lie in the $xz$-plane: ${\vec{a}_0}=c\hat{x}+s\hat{z}$,
${\vec{a}_1}=s\hat{x}-c\hat{z}$,
${\vec{b}_0}=\frac{1}{\sqrt{2}}(\hat{x}+\hat{z})$ and
${\vec{b}_1}=\frac{1}{\sqrt{2}}(-\hat{x}+\hat{z})$, the correlations have the following $3$-decomposition, 
\ba
P&=&\left(1-\mu-\nu\right) P^{\mathcal{G}=0}_{\mathcal{Q}=0}+\nu\left(\frac{P^{000}_{PR}+P^{11\gamma}_{PR}}{2}\right)+\mu P^{000}_{PR}, \label{BMSb1}
\ea
where the PR-box and Mermin box components, $\mu$ and $\nu$, 
are the same as for the correlations in Eq. (\ref{BMSb}).
%where $\mathcal{G}'$ and $\mathcal{Q}'$ are the irreducible PR-box and Mermin-box components which have the same amount as for the correlation in Eq. (\ref{BMSb}).
The $\mathcal{G}=\mathcal{Q}=0$ box, $P^{\mathcal{G}=0}_{\mathcal{Q}=0}$, in Eq. (\ref{BMSb1}) has nonmaximally mixed marginals, whereas the 
$\mathcal{G}=\mathcal{Q}=0$ box in Eq. (\ref{BMSb})
has maximally mixed marginals. Thus, the correlations in Eqs. (\ref{BMSb}) and (\ref{BMSb1}) differ only by their marginals because of this reason
the violation of the Bell-CHSH inequality is larger for the latter correlation than the former correlation (see fig. \ref{pt4}).
\subsection{Mixed quantum discordant states} 
Consider the correlations arising from the Werner states,
\be
\rho_W=p\ketbra{\psi^+}{\psi^+}+(1-p)\frac{\openone}{4}. 
\ee
which are entangled iff $p>\frac{1}{3}$ \cite{Werner}. It is known that 
the Werner states have nonzero quantum discord if $p>0$ \cite{WZQD}, similarly, we show that the Werner states can have Bell discord and Mermin discord 
if $p>0$. %Since the Werner states have maximally mixed marginals,
%they cannot give rise to nonmaximally mixed marginals correlations.
%can violate a Bell-CHSH inequality or 
%an EPR-steering inequality iff $p>\frac{1}{\sqrt{2}}$. 
\subsubsection{Bell-Werner box}
For the settings that gives rise to the optimal Bell-Schmidt box given in Eq. (\ref{BSb}), the correlations can be decomposed into the Tsirelson box and white noise
as follows, 
\be
P=p\left[\frac{1}{\sqrt{2}} P^{000}_{PR}+\left(1-\frac{1}{\sqrt{2}}\right)P_N\right]+(1-p)P_N. \label{BW}
\ee
%which has $\mathcal{G}=2\sqrt{2}p$. 
These correlations violate the Bell-CHSH inequality if $p>\frac{1}{\sqrt{2}}$ and have Bell discord $\mathcal{G}=2\sqrt{2}p>0$ if $p>0$. 
Bell discord of the local correlations in Eq. (\ref{BW}) reveals nonclassicality of the entangled states which cannot give rise to the violation of a 
Bell-CHSH inequality and separable nonzero quantum discord states
originating from nonlocality.
\subsubsection{Mermin-Werner box}
For the settings that gives rise to the optimal Mermin-Schmidt box given in Eq. (\ref{MSb}), the correlations can be decomposed into Mermin box and white noise
as follows,
\be
P=(1-p) P_N+ p\left(\frac{P^{000}_{PR}+P^{111}_{PR}}{2}\right). \label{MW}
\ee 
%which has $\mathcal{Q}=2p$.
These correlations violate the EPR-steering inequality if $p>\frac{1}{\sqrt{2}}$ and 
have Mermin discord $\mathcal{Q}=2p>0$ if $p>0$. 
Mermin discord of the local correlations in Eq. (\ref{MW}) reveals nonclassicality of the entangled states which cannot give rise to the violation of an 
EPR-steering inequality and separable nonzero quantum discord states 
originating from EPR-steering.
\subsubsection{Bell-Mermin-Werner box}
For the settings
${\vec{a}_0}=\sqrt{p}\hat{x}+\sqrt{1-p}\hat{y}$,
${\vec{a}_1}=\sqrt{1-p}\hat{x}-\sqrt{p}\hat{y}$,
${\vec{b}_0}=\frac{1}{\sqrt{2}}(\hat{x}+\hat{y})$ and
${\vec{b}_1}=\frac{1}{\sqrt{2}}(\hat{x}-\hat{y})$, the Werner states give rise to 
Bell-Mermin box which has Bell and Mermin discord simultaneously as follows, 
\be
\mathcal{G}=2\sqrt{2}p|\sqrt{p}-\sqrt{1-p}|>0 \quad \text{except when} \quad p\ne0, \frac{1}{2} \nonumber
\ee
and 
\ba
\mathcal{Q}&=&\sqrt{2}p|\sqrt{p}+\sqrt{1-p}-\left|\sqrt{p}-\sqrt{1-p}\right||\nonumber\\
&>&0 \quad \text{except when} \quad p\ne0, 1 \nonumber\\
&=&\left\{\begin{array}{lr}
2\sqrt{2}p\sqrt{p} \quad \text{when} \quad 0\le p\le\frac{1}{2}\\ 
2\sqrt{2}p\sqrt{1-p} \quad \text{when} \quad \frac{1}{2}\le p\le1.\\ 
\end{array}
\right.\nonumber
\ea
The correlations have a $3$-decomposition as follows:
\be
P=(1-q-r) P_N+ q\left(\frac{P^{000}_{PR}+P^{11\gamma}_{PR}}{2}\right)+rP^{000}_{PR}, \label{BMW}
\ee
where $q=\frac{p}{\sqrt{2}}|\sqrt{p}+\sqrt{1-p}-\left|\sqrt{p}-\sqrt{1-p}\right||$ and $r=\frac{p}{\sqrt{2}}\left|\sqrt{p}-\sqrt{1-p}\right|$.
Since the settings becomes optimal for Bell and Mermin discord when $p=1$ and $p=\frac{1}{2}$ respectively, the correlation in Eq. (\ref{BMW}) becomes
the Bell-Werner box in Eq. (\ref{BW}) at $p=1$ and the Mermin-Werner box in Eq. (\ref{MW}) at $p=\frac{1}{2}$.

It has been shown that the quantum correlations quantified by quantum discord in mixed states plays the role of entanglement in pure states and 
the Werner states are maximally quantum-correlated states \cite{QQC}.
Similarly, we observe that the correlations arising from the Werner states
in Eqs. (\ref{BW})-(\ref{BMW}) and the correlations arising from the Schmidt states in 
Eqs. (\ref{BSb}), (\ref{MSb}) and (\ref{BMSb}) have the similar decompositions; 
the parameter $p$ (quantum discord) in the Werner states plays the same role as the parameter $s$ (entanglement) in the Schmidt states.  
\subsection{Mixed nonmaximally entangled states}
We consider the correlations arising from the mixed states which can be written as a mixture of the Bell state and the classically-correlated state,
\be
\rho= p \ketbra{\psi^+}{\psi^+}+(1-p) \rho_{CC}, \label{rCC}
\ee
where $\rho_{CC}=\frac{1}{2}(\ketbra{00}{00}+\ketbra{11}{11})$. 
We illustrate that for the measurements that gives rise to optimal Bell discord, these states have the same behavior as the Werner states, 
and, for the measurements that gives rise to optimal Bell nonlocality, these states and the Schmidt states in Eq. (\ref{Schmidt}) 
have similar behavior:

 For the settings that gives rise to the noisy PR-box in Eq. (\ref{BSb}),
the correlations arising from the states in Eq. (\ref{rCC}) have the same decomposition as for the Bell-Werner box in Eq. (\ref{BW}) because the classically-correlated
state in Eq. (\ref{rCC}) gives rise to white noise for this settings which implies that 
the correlations violate the Bell-CHSH inequality if $p>\frac{1}{\sqrt{2}}$ and have Bell discord $\mathcal{G}=2\sqrt{2}p>0$ if 
$p>0$.

 For the settings
${\vec{a}_0}=\hat{z}$, ${\vec{a}_1}=\hat{x}$,
${\vec{b}_0}=\cos t\hat{z}+\sin t\hat{x}$ and ${\vec{b}_1}=\cos t\hat{z}-\sin t\hat{x}$, where $\cos t=\frac{1}{\sqrt{1+p^2}}$, the correlations arising
from the mixed states in Eq. (\ref{rCC})
violate the Bell-CHSH inequality i.e., $\mathcal{B}_{000}=2\sqrt{1+p^2}>2$ if $p>0$ and have Bell discord
$\mathcal{G}=\frac{4p^2}{\sqrt{1+p^2}}$. Thus, the correlations have analogous properties of the Bell-Schmidt box in Eq. (\ref{PRQ});
the parameter $p$ in the mixed entangled states plays the role of the parameter $s$ in the Schmidt states.
\subsection{Classical-quantum and quantum-classical states}
Here we show that all classical-quantum (CQ) and quantum-classical (QC) states have $\mathcal{G}=\mathcal{Q}=0$ for all measurement settings. 
The CQ states can be written as,
\be 
\rho_{CQ}=\sum^{1}_{i=0}p_i\ketbra{i}{i}\otimes \chi_i \label{c-q},
\ee
whereas QC states can be written as, 
\be
\rho_{QC}=\sum^1_{j=0}p_j\phi_j\otimes \ketbra{j}{j}\label{q-c},
\ee 
where $\{\ket{i}\}$ and $\{\ket{j}\}$ are the orthonormal sets and $\chi_i$ and  $\phi_j$ are the quantum states. 
Despite the CQ and QC states are not the product states in general, their joint expectation value can be written in the factorized form, 
$\braket{AB}= f(\hat{a})f(\hat{b})$, here $\hat{a}$ and $\hat{b}$ are the measurement settings chosen by Alice and Bob respectively.
This factorization of the expectation values for the CQ and QC states 
implies that they cannot have nonzero Bell discord or nonzero Mermin discord for all measurements.
\begin{proof}
In the Bloch sphere representation, the CQ
states in Eq. (\ref{c-q}) can be written as:
\begin{eqnarray}
\rho_{CQ}&=&\frac{p_0}{4}\left(\openone
+\hat{r}\cdot\vec{\sigma}\right)\otimes\left(\openone
+\vec{s}_0\cdot\vec{\sigma}\right)\nonumber \\&&+\frac{p_1}{4}\left(\openone
-\hat{r}\cdot\vec{\sigma}\right)\otimes\left(\openone
+\vec{s}_1\cdot\vec{\sigma}\right),
\end{eqnarray}
where $\hat{r}$ is the Bloch vector for the projectors $\ketbra{i}{i}$ and
$\vec{s}_i$ are the Bloch vector for the states $\chi_i$. Notice that $\hat{r}$ appears twice in the above decomposition because of the orthogonality of 
projectors on Alice's side; as a result of this, the expectation value factorizes as follows,
\be
\braket{A_iB_j}=\left(\hat{a_i}\cdot\hat{r}\right) \left(\hat{b_j}\cdot(p_0\vec{s}_0-p_1\vec{s}_1)\right), \label{product-like}
\ee
whose form is similar to that of a product state, $\rho=\rho_A\otimes\rho_B=\frac{1}{4}\left[\left(\openone
+\vec{r}\cdot\vec{\sigma}\right)\otimes\left(\openone
+\vec{s}\cdot\vec{\sigma}\right)\right]$. 
The analysis in the previous sections shows that the optimal settings have the following property: 
for the Bell discord one has, $\hat{a}_0\cdot\hat{a}_1=0$, $\hat{b}_0\cdot\hat{b}_1=0$ and $\hat{a}_i\cdot\hat{b}_j=\pm\frac{1}{\sqrt{2}}$, whereas
for the Mermin discord one has: $\hat{a}_0\cdot\hat{a}_1=0$, $\hat{b}_0\cdot\hat{b}_1=0$ and $\hat{a}_i=\pm\hat{b}_j$. 
Since the optimal settings that maximizes $\mathcal{G}$ and $\mathcal{Q}$ have the common property that measurements on
Alice's side or Bob's side are orthogonal, we choose orthogonal measurements on Alice' side
to maximize $\mathcal{G}$ and $\mathcal{Q}$ with respect to the correlation given in Eq. (\ref{product-like}).
Suppose we choose $\hat{a}_0 \cdot \hat{r}=1$, the orthogonality  
condition ($\hat{a}_0\cdot\hat{a}_1=0$) implies that $\hat{a}_1 \cdot \hat{r}=0$. For this choice of orthogonal measurements on Alice's side,   
$\mathcal{B}_{00}=|(\hat{b_0}+\hat{b_1})\cdot(p_0\vec{s_0}-p_1\vec{s_1})|$,
$\mathcal{B}_{01}=|(\hat{b_0}-\hat{b_1})\cdot(p_0\vec{s_0}-p_1\vec{s_1})|$,
$\mathcal{B}_{10}=|(\hat{b_0}+\hat{b_1})\cdot(p_0\vec{s_0}-p_1\vec{s_1})|$, and 
$\mathcal{B}_{11}=|(\hat{b_0}-\hat{b_1})\cdot(p_0\vec{s_0}-p_1\vec{s_1})|$ which
implies that $\mathcal{G}=\mathcal{Q}=0$ for all possible measurements on Bob's side. Similarly, we can prove that $\mathcal{G}=\mathcal{Q}=0$ 
for the QC states since $\mathcal{G}$ and $\mathcal{Q}$ are symmetric under the parties permutation.
\end{proof}
Since the joint expectation value of any quantum-correlated state, which is neither a CQ state nor a QC state, cannot be written in the factorized form 
i.e., $\braket{AB}\ne f(\hat{a})f(\hat{b})$, 
all quantum correlated states which have nonzero left and right quantum discord can give rise to nonzero Bell discord or/and Mermin discord.
\subsection{Tsirelson bound}
\begin{figure}\label{Tsirelson}
\centering
\includegraphics[width=0.37\textwidth]{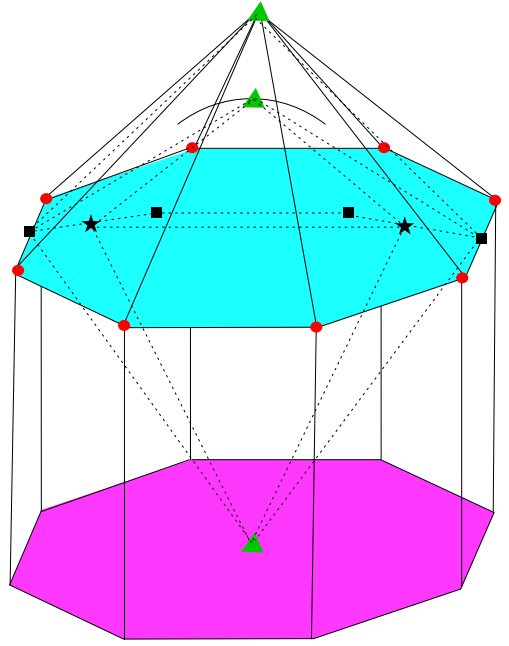} 
\caption{The square and the star points on the facet of the local polytope represent the classically-correlated (CC) boxes  and the quantum Mermin boxes
respectively. 
The subpolytope, $\mathcal{N}_{mm}$, 
formed by the PR-boxes and the CC boxes is represented by the region connecting 
the triangle point on the top, the square points and the triangle point at the centre of the bottom (white noise). 
%represents  %which is obtained by constraining the full NS polytope by maximally mixed marginat ls.
The subpolytope, $\mathcal{N}_{Tmm}$, whose vertices are the Tsirelson boxes and CC boxes is represented by the region connecting 
the triangle point on the curved surface,
the square points and white noise. The subpolytope, $\mathcal{N}_{Q}$, whose vertices are the Tsirelson boxes and Mermin boxes is represented
by the region connecting the triangle point on the curved surface,
the star points and white noise.
The region connecting the square points and white noise represents the subpolytope, $\mathcal{L}_{mm}$, formed by the CC boxes.
The subpolytope, $\mathcal{L}_{Q}$, formed by the Mermin boxes is represented by the region connecting the star points and white noise.}\label{finalfig}
\end{figure} 
Here we are interested in a restricted NS polytope, $\mathcal{N}_{Q}$, whose vertices are the $8$ Tsirelson boxes,
\be
P^{\alpha\beta\gamma}_{T}=\frac{1}{\sqrt{2}}P^{\alpha\beta\gamma}_{PR}+\left(1-\frac{1}{\sqrt{2}}\right)P_N, \label{tbox}
\ee
and the $8$ quantum Mermin-boxes, $P^{\alpha\beta\gamma}_M$, which are given in Eq. (\ref{Mmmm}) to figure out the constraints of quantum correlations.
This polytope can be realized by quantum theory  which we illustrate by the correlations arising from the convex mixture of 
the $8$ maximally entangled states,
\be
\rho=\sum^1_{k=0}\sum^1_{j=0}p^j_k \ketbra{\psi^j_k}{\psi^j_k} +\sum^1_{k=0}\sum^1_{j=0} q^j_k\ketbra{\phi^j_k}{\phi^j_k}, \label{ch2bds}
\ee
where $\ket{\psi^j_k}=\frac{1}{\sqrt{2}}(\ket{00}+(-1)^{j} i^{k}\ket{11})$ and $\ket{\phi^j_k}=\frac{1}{\sqrt{2}}(\ket{01}+(-1)^{j} i^{k}\ket{10})$.
For the measurement settings, $\mathcal{M}_{T}$:
\be
{\vec{a}_0}=\hat{x}, \quad {\vec{a}_1}=\hat{y},\quad
{\vec{b}_0}=\frac{1}{\sqrt{2}}(\hat{x}-\hat{y}) \quad \text{and} \quad {\vec{b}_1}=\frac{1}{\sqrt{2}}(\hat{x}+\hat{y}), \label{M_N}
\ee 
the correlation arising from the states in Eq. (\ref{ch2bds}) can be decomposed into $8$ Tsirelson boxes,
\ba
P(\rho,\mathcal{M}_{T})\!&=&\!p^0_0P^{000}_{T}+p^1_0P^{001}_{T}+p^0_1P^{100}_{T}+p^1_1P^{101}_{T}\nonumber\\
&&+q^0_0P^{011}_{T}+q^1_0P^{010}_{T}+q^0_1P^{111}_{T}+q^1_1P^{110}_{T}.\label{8Tb}
\ea
For the measurement settings, $\mathcal{M}_{M}$:   
\be
{\vec{a}_0}=\hat{x}, \quad {\vec{a}_1}=\hat{y}, \quad 
{\vec{b}_0}=-\hat{y} \quad \text{and} \quad {\vec{b}_1}=\hat{x},\label{M_C}
\ee 
the correlation arising from the states in Eq. (\ref{ch2bds}) can be decomposed into $8$ Mermin boxes,
\ba
P(\rho,\mathcal{M}_{M})\!&=&\!p^0_0P^{000}_{M}+p^1_0P^{001}_{M}+p^0_1P^{100}_{M}+p^1_1P^{101}_{M}\nonumber\\
&&+q^0_0P^{011}_{M}+q^1_0P^{010}_{M}+q^0_1P^{111}_{M}+q^1_1P^{110}_{M}.\label{8Mb}
\ea
Since the set of quantum correlations is convex \cite{Pitowski,WernerWolfmulti},
any convex mixture of the two correlations given in Eqs. (\ref{8Tb}) and (\ref{8Mb}),
\be
P=\lambda P(\rho,\mathcal{M}_{T})+(1-\lambda)P(\rho,\mathcal{M}_{M}),
\ee
is also quantum realizable which implies that the polytope $\mathcal{N}_{Q}$
is quantum. %Further, $\mathcal{N}_{Q}$ is the largest polytope with minimum number of classical boxes within the convex body of quantum boxes.

We obtain the following relationship between the two quantum correlations given in Eqs. (\ref{8Tb}) and (\ref{8Mb}). 
\begin{observation}
For any state given in Eq. (\ref{ch2bds}), Bell discord of the correlation given in Eq. (\ref{8Tb}) is related to the Mermin discord of the correlation 
given in Eq. (\ref{8Mb}) as follows,
\be
\mathcal{G}(P(\rho, \mathcal{M}_T))=\sqrt{2}\mathcal{Q}(P(\rho, \mathcal{M}_{M})). \label{rBMD}
\ee
\end{observation}
\begin{proof}
The Bell functions for the settings given in Eq. (\ref{M_N}) reduce to the Mermin functions for the settings given in Eq. (\ref{M_C}) as follows: 
\begin{widetext}
\ba
\mathcal{B}_{\alpha\beta}&\!=\!&\frac{1}{\sqrt{2}}|\braket{\sigma_x\otimes (\sigma_x+\sigma_y)}+(-1)^\beta\braket{\sigma_x\otimes (\sigma_x-\sigma_y)}
+(-1)^\alpha\braket{\sigma_y\otimes (\sigma_x+\sigma_y)}\nonumber\\
&&+(-1)^{\alpha\oplus\beta\oplus1}\braket{\sigma_y\otimes (\sigma_x-\sigma_y)}|\nonumber\\
&=&\left\{
\begin{array}{lr}
(\alpha\oplus\beta\oplus1)\sqrt{2}|(-1)^{\beta}\braket{\sigma_x\otimes \sigma_x}\!+\!(-1)^{\alpha}\braket{\sigma_y\otimes \sigma_y}| 
\\+(\alpha\oplus\beta)\sqrt{2}|(-1)^{\gamma}\braket{\sigma_x\otimes \sigma_y}+(-1)^{\alpha\oplus\beta\oplus\gamma\oplus1}\braket{\sigma_y\otimes \sigma_x}| \\
=\sqrt{2}\mathcal{M}_{\alpha\beta}     \quad \text{for} \quad \alpha\beta=00, 01,\\
(\alpha\oplus\beta)\sqrt{2}|(-1)^{\beta}\braket{\sigma_x\otimes \sigma_x}\!+\!(-1)^{\alpha}\braket{\sigma_y\otimes \sigma_y}|\\
+(\alpha\oplus\beta\oplus1)\sqrt{2}|(-1)^{\gamma
}\braket{\sigma_x\otimes \sigma_y}+(-1)^{\alpha\oplus\beta\oplus\gamma\oplus1}\braket{\sigma_y\otimes \sigma_x}|  \\
=\sqrt{2}\mathcal{M}_{\alpha\beta}    \quad\! \!\text{for} \quad \alpha\beta=10,11
\end{array}
\right.
 \label{B-M}
\ea
\end{widetext}
due to the linearity of quantum theory, $\braket{A+B}=\braket{A}+\braket{B}$.
The relationship between the Bell and Mermin functions given in Eq. (\ref{B-M}) implies that 
$\mathcal{G}(\rho, \mathcal{M}_T)=\sqrt{2}\mathcal{Q}(\rho, \mathcal{M}_{M})$.
\end{proof} 
The relationship between Bell and Mermin discord given in Eq. (\ref{rBMD}) implies that the Mermin boxes  
limit nonlocality of the most nonlocal quantum boxes
to the Tsirelson bound since $\mathcal{G}(\rho, \mathcal{M}_T)\le 2\sqrt{2}$ follows from the fact that $\mathcal{Q}(\rho, \mathcal{M}_M)\le2$.

We now discuss the constraints of the quantum region, $\mathcal{N}_{Q}$, inside the full NS polytope. 
Notice that correlations in the region $\mathcal{N}_{Q}$ have maximal local randomness i.e., $\braket{A}_i=\braket{B}_j=0$. If the full NS polytope is constrained by
maximal local randomness, it gives rise to a subpolytope, $\mathcal{N}_{mm}$, whose vertices are the $8$ PR-boxes and $8$ classically-correlated (CC) boxes,
\be
P_{CC}^{\alpha\beta\gamma}(a_m,b_n|A_i,B_j)=\left\{
\begin{array}{lr}
\frac{1}{2}, & m\oplus n=\alpha i \oplus \beta j \oplus \gamma \\ 
0 , & \text{otherwise}.\\
\end{array} \label{CCE}
\right.
\ee  
The polytope, $\mathcal{N}_{Tmm}$, whose vertices are the $8$ Tsirelson boxes and the $8$ CC boxes is obtained by constraining $\mathcal{N}_{mm}$ by the 
Tsirelson inequalities, $\mathcal{B}_{\alpha\beta\gamma}\le2\sqrt{2}$ \cite{tsi1}. The polytope $\mathcal{N}_{Tmm}$ is quantum since its vertices are quantum realizable \cite{Pitowski}. 
Notice that the polytope, $\mathcal{N}_Q$, is contained inside $\mathcal{N}_{Tmm}$ (see fig. \ref{finalfig}).  
Since the Mermin boxes with maximally mixed marginals limits nonlocality of quantum correlations, finding the physical constraints of
$\mathcal{N}_Q$ would help us to single out quantum theory.
The set of local boxes which have maximal local randomness forms a polytope, $\mathcal{L}_{mm}$,
whose vertices are the CC boxes.
Inside this polytope, there exists a polytope, $\mathcal{L}_{Q}$, whose vertices are the $8$ maximally mixed marginals Mermin boxes. 
%Since this Mermin-box polytope is purely a quantum region and its vertices are unique to quantum theory, finding the constraints of this polytope would help
%us to single out quantum theory.  
\section{Conclusions}\label{conc}
We have introduced the measures, Bell discord ($\mathcal{G}$) and Mermin discord ($\mathcal{Q}$), to characterize bipartite quantum correlations within the framework of GNST. 
We find that the full NS polytope can be divided into four parts: (i) $\mathcal{G}>0$ $\&$ $\mathcal{Q}=0$ region (ii) $\mathcal{G}=0$ $\&$ $\mathcal{Q}>0$ region 
(iii) $\mathcal{G}>0$ $\&$ $\mathcal{Q}>0$ region and (iv) $\mathcal{G}=\mathcal{Q}=0$ polytope. 
By using this division, we have obtained the $3$-decomposition of any NS box 
into PR-box, a maximally-local box with $\mathcal{Q}=2$ and a
$\mathcal{G}=\mathcal{Q}=0$ box. We have introduced two types of Mermin boxes which are local and extremal with respect to the $3$-decomposition.
We have identified the largest quantum region with minimal $\mathcal{G}=\mathcal{Q}=0$ boxes that is a subpolytope of the full NS polytope. 
This subpolytope gives us insights to find out what singles out
quantum theory from other nonsignaling theories.

We have applied Bell discord and Mermin discord to quantify nonclassicality of quantum correlations arising from the $2 \times 2$ (two-qubit) states. 
We find that the quantum-correlated states  which are neither classical-quantum states nor quantum-classical states
can give rise a $3$-decomposition i.e., nonzero Bell discord or/and Mermin discord for suitable incompatible measurements.
We find that when pure entangled states and Werner states give rise optimal Bell or Mermin discord, 
quantum correlations quantified by quantum discord in the Werner states
plays role analogous to entanglement in the pure states.
However, we have considered only those boxes with two binary inputs and two binary outputs.
Similarly, it would be interesting to study quantum correlations arising from $d_A \times d_B$ states  
by using NS polytope in which the black boxes have more inputs and more outputs \cite{Barrett,minput}.

It is known that the Bell inequalities serve as device-independent witnesses for entanglement i.e., 
it reveals the presence of entanglement in the observed statistics without any knowledge
of the dimension of the measured system and the measurement devices used \cite{DQKD}.
A nonzero Bell or Mermin discord of a local box cannot be a device-independent witness of nonclassicality
i.e., nonzero quantum discord and incompatibility of measurements.
This follows from the fact 
that local correlations which have nonzero Bell or Mermin discord can also arise from the separable classical 
states in higher dimensional space; for instance, a Mermin box can also arise from a classically-correlated state in higher 
dimensional space for compatible measurements \cite{NLRan}.
However, Bell and Mermin discord are 
semi-device-independent witnesses of nonclassicality of local boxes \cite{SDI} in the sense that they reveal the presence of nonzero quantum discord 
and incompatibility of measurements when the dimension of the measured system is restricted to be $2 \times 2$ and without 
any knowledge of the measurement devices used.

%It would be interesting to generalize Bell discord to
%study nonclassicality of quantum correlations arising from 
%$d \times d$ states using $d$-box polytope in which the black boxes have $d$ outputs.
\section*{Acknowledgements}
I thank IISER Mohali for financial support. I am very thankful to Dr. P. Rungta for 
suggesting this problem and for many inspiring discussions subsequently.
I thank Dr. R. Srikanth (PPISR Bangalore), Manik Banik (ISI Kolkata) and Dr. K. P. Yogendran (IISER Mohali) for discussions. 
I am grateful to IJQI and two anonymous referees for comments and suggestions.   
\appendix
\section{An example to illustrate the notion of irreducible PR-box in the decomposition}\label{irreducible}
Here we illustrate the relation between the minimization in the definition of Bell discord given in Eq. (\ref{defBD}) and the irreducible PR-box component in the 
canonical decomposition given in Eq. (\ref{Gde}) with the following example:
\be
P=0.4P^{000}_{PR}+0.3P^{010}_{PR}+0.2P^{100}_{PR}+0.1P^{110}_{PR}. \label{xxxx}
\ee
We are now interested in finding whether this correlation has the irreducible PR-box component as it has single PR-box excess, $P^{000}_{PR}$. 
Let us now consider the other possible decompositions for the above correlation in which it can be written as a convex sum of the reduced PR-box and the three local boxes
which are uniform mixture of the two PR-boxes.  
If we start reducing the PR-box excess with the last PR-box to 
a convex sum of the irreducible PR-box and the uniform mixture of the two PR-boxes, 
the resulting decomposition obtained by further reducing this reduced PR-box excess with the other two PR-boxes 
can be written as the convex sum of a 
reduced PR-box and a local box,
\be
P=\mu P^{000}_{PR}+(1-\mu) P_L,
\ee
where $\mu=0.2$ and $P_L=\frac{1}{8}P^{000}_{PR}+\frac{1}{2}P^{010}_{PR}+\frac{1}{4}P^{100}_{PR}+\frac{1}{8}P^{110}_{PR}$. 
The nonzero reduced PR-box excess, $\mu$, in this decomposition, however, 
cannot be irreducible as $\mu$ vanishes for the other possible decompositions; 
for instance, if we start to reduce the PR-box excess with the second PR-box, $\mu$ will vanish. For the correlation in Eq. (\ref{xxxx}), 
$\mathcal{G}_1=\mathcal{G}_2=0$ and $\mathcal{G}_3=0.8$ 
which explains why the minimization in Eq. (\ref{defBD}) is required as $\mathcal{G}$ is intended to quantify irreducible PR-box component. 
\section{Linearity of Bell and Mermin discord w.r.t the canonical decompositions}\label{lbdmd}
$\mathcal{G}$ is, in general, not linear for the decomposition of a given correlation into the convex mixture of two $\mathcal{G}>0$ boxes.
For instance, consider a correlation which is the convex mixture of two PR-boxes,
\be
P=pP^i_{PR}+qP^j_{PR}; \quad p>q,\label{nlg>}
\ee
which has $\mathcal{G}(P)=4(p-q)$. Suppose $\mathcal{G}$ is linear for this decomposition, $\mathcal{G}(P)=p\mathcal{G}(P^i_{PR})+q\mathcal{G}(P^j_{PR})=4\ne4(p-q)$.
However, $\mathcal{G}$ is linear for the decomposition of the correlation in Eq. (\ref{nlg>}) into a mixture of a single PR-box and a $\mathcal{G}=0$ box,
\be
P=(p-q)P^i_{PR}+2q\left(\frac{P^i_{PR}+P^j_{PR}}{2}\right). \label{ggnl}
\ee
$\mathcal{G}$ is, in general, also not linear for the decomposition of a correlation into the convex mixture of two $\mathcal{G}=0$ boxes.
For instance, consider the following uniform mixture of two Mermin boxes (the triangle point on the facet of the local polytope in fig. \ref{NS3dfig}),
\be
P=\frac{1}{2}P^{1}_M+\frac{1}{2}P^{2}_M,\label{nggnl}
\ee
where $P^{1}_M=\frac{1}{2}\left(P^{000}_{PR}+P^{111}_{PR}\right)$ and $P^{2}_M=\frac{1}{2}\left(P^{000}_{PR}+P^{110}_{PR}\right)$. Evaluation of $\mathcal{G}$
on the right hand side by using linearity gives $\frac{1}{2}\mathcal{G}(P^{1}_M)+\frac{1}{2}\mathcal{G}(P^{2}_M)=0$, however, $\mathcal{G}(P)=2$.
The correlation in Eq. (\ref{nggnl}) can also be written in the isotropic PR-box form as follows,
\be
P=\frac{1}{2}P^{000}_{PR}+\frac{1}{2}P_N.
\ee
It is obvious that $\mathcal{G}$ is linear for this decomposition. Similarly, we can observe that Mermin discord is, in general, not linear for the 
the decomposition of a given correlation into a mixture of two $\mathcal{Q}>0$ boxes or $\mathcal{Q}=0$ boxes and linear for the canonical decomposition.


\begin{thebibliography}{99}
\bibitem{bell64} 
J.S. Bell, Physics. {\bf 1}, 195 (1964).
\bibitem{BNL}N. Brunner, D. Cavalcanti, S. Pironio, V. Scarani, and S. Wehner, Rev. Mod. Phys. \textbf{86}, 419 (2014).
\bibitem{tsi1}
B.~S. Tsirel'son, Lett. Math. Phys. \textbf{4}, 93 (1980).
\bibitem{PR} S. Popescu, D. Rohrlich, Found. Phys. {\bf 24}, 379 (1994).
\bibitem{Barrett} J. Barrett, N. Linden, S. Massar, S. Pironio, S. Popescu, and D. Roberts, Phys. Rev. A {\bf 71}, 022101 (2005).
\bibitem{MAG06}Ll. Masanes, A. Acin, and N. Gisin, Phys. Rev. A {\bf 73}, 012112 (2006).
\bibitem{Pitowski}I. Pitwosky, preprint (2001), arXiv:quant-ph/0112068.
\bibitem{Geb}P. Skrzypczyk, N. Brunner, and S. Popescu, Phys. Rev. Lett. {\bf 102}, 110402 (2009).
\bibitem{DQKD} A. Ac\'{i}n, N. Gisin, and L. Masanes, Phys. Rev. Lett. {\bf 97}, (2006).
\bibitem{GT}N. Gisin, Phys. Lett. A 154, {\bf 201} (1991).
\bibitem{PRQB}S. Popescu and D. Rohrlich, Phys. Lett. A {\bf 166}, 293 (1992).
\bibitem{Werner}R. F. Werner, Phys. Rev. A. {\bf 40}, 4277 (1989). 
\bibitem{IncomN}M. T. Quintino, T. V\'{e}rtesi, and N. Brunner, Phys. Rev. Lett. \textbf{113}, 160402 (2014).
\bibitem{Grudkaetal}A. Grudka \etal, {\it Phys. Rev. Lett} {\bf 112} (2014) 120401.
\bibitem{KS}S. Kochen and E. P. Specker, J. Math. Mech. \textbf{17}, 59 (1967).
\bibitem{Peres} A. Peres, Phys. Lett. A, {\bf 151}, 107 (1990).
\bibitem{UNLH}N. D. Mermin, Phys. Rev. Lett. {\bf 65}, 3373 (1990).
\bibitem{WZQD} H. Ollivier and W. H. Zurek, Phys. Rev. Lett. {\bf 88}, 017901 (2001).
\bibitem{Dakicetal}Borivoje Dakic, Vlatko Vedral, and Caslav Brukner, Phys. Rev. Lett. {\bf 105},
190502 (2010).
\bibitem{Jeba}C. Jebarathinam, arXiv:1410.1472 (2014).
\bibitem{chsh} J.~F. Clauser, M.~A. Horne, A. Shimony and R.~A. Holt, Phys. Rev. Lett. \textbf{23}, 880 (1969).
\bibitem{mermin} N.D. Mermin, Phys. Rev. Lett. {\bf 65}, 1838 (1990).
%\bibitem{UNLH}N. D. Mermin, Phys. Rev. Lett. {\bf 65}, 3373 (1990).
%\bibitem{Peres} A. Peres, Phys. Lett. A, {\bf 151}, 107 (1990).
\bibitem{QQC}S. M. Giampaolo, A. Streltsov, W. Roga, D. Bru\ss{}, F. Illuminati, {\it Phys. Rev. A} {\bf 87}, (2013) 012313.
\bibitem{EPRsi}D. J. Saunders, S. J. Jones, H. M. Wiseman and G. J. Pryde, {\it Nature Physics} {\bf 6} (2010) 845.
\bibitem{CJWR} E. G. Cavalcanti, S. J. Jones, H. M. Wiseman, and M. D. Reid, Phys. Rev. A \textbf{80}, 032112 (2009).
\bibitem{Fine}A. Fine, Phys. Rev. Lett. \textbf{48}, 291 (1982).
%\bibitem{BNL}N. Brunner, D. Cavalcanti, S. Pironio, V. Scarani, and S. Wehner, Rev. Mod. Phys. \textbf{86}, 419 (2014).
%\bibitem{IncomN}M. T. Quintino, T. V\'{e}rtesi, and N. Brunner, Phys. Rev. Lett. \textbf{113}, 160402 (2014).
\bibitem{WernerWolf}R. F. Werner, and M. M. Wolf, Quantum Inf. Comput. \textbf{1}, 1 (2001).
%\bibitem{Sch}A. Peres, Phys. Lett. A. {\bf 202}  16 (1995).
%\bibitem{MAG06}Ll. Masanes, A. Acin, and N. Gisin, Phys. Rev. A {\bf 73}, 012112 (2006).
%\bibitem{EPR2} A. Elitzur, S. Popescu, and D. Rohrlich, Phys. Lett. A \textbf{162}, 25 (1992)
\bibitem{Bellmono}M. Pawlowski and C. Brukner, {\it Phys. Rev. Lett} {\bf 102} (2009) 030403.
\bibitem{Short}A. J. Short, {\it Phys. Rev. Lett} \textbf{102} (2009) 180502.
\bibitem{Bub} J. Bub, arXiv:1210.6371 (2014).
\bibitem{WernerWolfmulti}R. F. Werner, and M. M. Wolf, {\it Phys. Rev. A} {\bf 64} (2001) 032112.
\bibitem{KCK}P. Kurzy\ifmmode \acute{n}\else \'{n}\fi{}ski, A. Cabello and D. Kaszlikowski, {\it Phys. Rev. Lett.} {\bf 112} (2014) 100401.
%\bibitem{Poly}G. M. Ziegler, Lectures on Polytopes, Graduate Texts in Mathematics, Vol. 152 (Springer, Berlin, 1995).
\bibitem{QCall}A. Ferraro, L. Aolita, D. Cavalcanti, F. M. Cucchietti and A. Acin, Phys. Rev. A {\bf 81}, 052318 (2010).
\bibitem{Caves}M. D. Lang and C. M. Caves, {\it Phys. Rev. Lett} {\bf 105} (2010) 150501.
\bibitem{Sch1}E. Schmidt, {\it Math. Ann} \textbf{63} (1907) 433.
\bibitem{Sch2}A. Peres, {\it Quantum theory: concepts and methods} (Kluwer Academic Publishers, Dordrecht 1995).
\bibitem{tangle}V. Coffman, J. Kundu, and W. K. Wootters, Phys. Rev. A {\bf 61}, 052306 (2000).
\bibitem{minput}J. Barrett and S. Pironio, {\it Phys. Rev. Lett.} {\bf 95} (2005) 140401. 
\bibitem{NLRan}A. Ac\'{i}n, S. Massar, and S. Pironio, Phys. Rev. Lett. {\bf 108}, 100402 (2012).
%\bibitem{Jeba}C. Jebarathinam, arXiv:1410.1472 (2014).
%\bibitem{NCI}A. Cabello, S. Filipp, H. Rauch and Y. Hasegawa, Phys. Rev. Lett. \textbf{100}, 130404 (2008).
%\bibitem{KS}S. Kochen and E. P. Specker, J. Math. Mech. \textbf{17}, 59 (1967).
\bibitem{SDI}Y.-C. Liang, T. V\'{e}rtesi, and N. Brunner, {\it Phys. Rev. A} \textbf{83} (2011) 022108.
\end{thebibliography}
\end{document}